\documentclass[titlepage=off,11pt,DIV=10]{scrartcl}

\usepackage[T1]{fontenc}
\usepackage[utf8]{inputenc}

\let\latextextsuperscript\textsuperscript
\AtBeginDocument{\let\textsuperscript\latextextsuperscript}

\usepackage{amsthm,amsmath,amsfonts,amssymb}
\usepackage{newtxtext,newtxmath,bbm,dsfont}
\input{braket}

\usepackage[dvipsnames]{xcolor}
\definecolor{color1}{RGB}{230,57,70}
\definecolor{color2}{RGB}{29,53,87}
\definecolor{color3}{RGB}{69,123,157}
\usepackage[colorlinks=true,linkcolor=color1,urlcolor=color2,citecolor=color3,anchorcolor=green,pdfusetitle]{hyperref}
\usepackage{xspace,xparse}
\usepackage{microtype}
\usepackage{tabularx,xifthen,afterpage,booktabs}
\usepackage{qcircuit,algpseudocode,algorithm}

\usepackage{tikz}
\usetikzlibrary{automata,positioning,arrows.meta}

\tikzset{
>=Latex,shorten >=2pt,shorten <=2pt,
node distance=3cm,
every state/.style={fill=gray!20},
initial text={\normalfont start}, 
}

\usepackage{cleveref}
\crefname{lemma}{lemma}{lemmas}
\crefname{proposition}{proposition}{propositions}
\crefname{definition}{definition}{definitions}
\crefname{theorem}{theorem}{theorems}
\crefname{conjecture}{conjecture}{conjectures}
\crefname{corollary}{corollary}{corollaries}
\crefname{example}{example}{examples}
\crefname{section}{section}{sections}
\crefname{appendix}{appendix}{appendices}
\crefname{figure}{fig.}{figs.}
\crefname{equation}{eq.}{eqs.}
\crefname{table}{table}{tables}
\crefname{item}{property}{properties}
\crefname{remark}{remark}{remarks}
\crefname{problem}{}{}

\newtheorem{theorem}{Theorem}
\newtheorem{definition}[theorem]{Definition}
\newtheorem{corollary}[theorem]{Corollary}

\newtheorem{lemma}[theorem]{Lemma}

\newcommand\prob\textsc
\makeatletter
\newcommand{\problemtitle}[1]{\gdef\@problemtitle{#1}}
\newcommand{\probleminput}[1]{\gdef\@probleminput{#1}}
\newcommand{\problemquestion}[1]{\gdef\@problemquestion{#1}}
\newcommand{\problempromise}[1]{\gdef\@problempromise{#1}}
\DeclareDocumentEnvironment{problem}{o o}{
\problemtitle{\IfValueT{#1}{#1}}\probleminput{}\problempromise{}\problemquestion{}
}{
  \par\addvspace{.5\baselineskip}
  \noindent
  \begin{tabularx}{\textwidth}{@{\hspace{\parindent}} l X c}
    \multicolumn{2}{@{\hspace{\parindent}}l}{\prob{\@problemtitle}} \\
    \textbf{Input:} & \@probleminput \\
\ifx\@problempromise\empty%
\else%
    \textbf{Promise:} & \@problempromise \\
\fi%
    \textbf{Question:} & \@problemquestion
  \end{tabularx}
  \par\addvspace{.5\baselineskip}
}
\makeatother

\usepackage{listings}
\newcommand{\C}{{\mathbb{C}}} 

\newcommand{\Exp}[2][]{\mathbb{E}{%
	\ifthenelse{\isempty{#1}}{}{_{#1}}}%
	\left[{#2}\right]} 

\renewcommand{\O}{\mathcal{O}}

\newcommand{\Abs}[1]{\left|{#1}\right|}

\newcommand{\cc}[1]{\mathsf{{#1}}}	

\DeclareMathOperator{\BigO}{\mathrm O}

\newcommand\field\mathds

\DeclareMathOperator{\PowerLaw}{Power}

\DeclareMathOperator{\argmax}{argmax}
\newcommand{\dd}{\mathrm d}
\renewcommand{\1}{\mathbbm 1}

\newcommand\op\mathbf

\DeclareMathOperator{\diag}{diag}

\usepackage[backend=biber,style=alphabetic,sorting=none,doi=false,isbn=false,url=false,maxbibnames=20,maxcitenames=2]{biblatex}
\bibliography{refs}
\newbibmacro{string+doi}[1]{\iffieldundef{doi}{#1}{\href{https://dx.doi.org/\thefield{doi}}{#1}}}
\DeclareFieldFormat{title}{\usebibmacro{string+doi}{\mkbibemph{#1}}}
\DeclareFieldFormat[article]{title}{\usebibmacro{string+doi}{\mkbibquote{#1}}}
\DeclareFieldFormat[incollection]{title}{\usebibmacro{string+doi}{\mkbibquote{#1}}}
\DeclareFieldFormat[inproceedings]{title}{\usebibmacro{string+doi}{\mkbibquote{#1}}}
\renewbibmacro{in:}{}

\newcommand\ee{\mathrm e}
\newcommand\ii{\mathbbm i}

\DeclareMathOperator{\supp}{supp}

\DeclareMathOperator{\Var}{Var}

\title{A Quantum Search Decoder\\for Natural Language Processing}
\date{June 2020}

\usepackage{authblk}

\author{%
  Johannes Bausch\\
  CQIF, DAMTP\\
  University of Cambridge\\
  Cambridge CB3 0WA \\
  \texttt{jkrb2@cam.ac.uk} \\
  \and
  Sathyawageeswar Subramanian \\
  CQIF, DAMTP\\
  University of Cambridge\\
  Cambridge CB3 0WA \\
  \texttt{ss2310@cam.ac.uk} \\
  \and
  Stephen Piddock \\
  School of Mathematics,
  University of Bristol,
  Bristol BS8 1TW\\
  Heilbronn Institute for Mathematical Research, Bristol\\
  \texttt{stephen.piddock@bristol.ac.uk}
}

\begin{document}
\maketitle
\thispagestyle{empty}
\enlargethispage{2cm}

\begin{abstract}
Probabilistic language models, e.g.\ those based on an LSTM, often face the problem of finding a high probability prediction from a sequence of random variables over a set of tokens. This is commonly addressed using a form of greedy decoding such as beam search, where a limited number of highest-likelihood paths (the beam width) of the decoder are kept, and at the end the maximum-likelihood path is chosen.
  
  In this work, we construct a quantum algorithm to find the globally optimal parse (i.e.\ for infinite beam width) with high constant success probability.
  When the input to the decoder is distributed as a power-law with exponent $k>0$, our algorithm has runtime $R^{n f(R,k)}$, where $R$ is the alphabet size, $n$ the input length; here $f<1/2$, and $f\rightarrow 0$ exponentially fast with increasing $k$, hence making our algorithm always more than quadratically faster than its classical counterpart.
  
  We further modify our procedure to recover a finite beam width variant, which enables an even stronger empirical speedup while still retaining higher accuracy than possible classically.
  Finally, we apply this quantum beam search decoder to Mozilla's implementation of Baidu's \emph{DeepSpeech} neural net, which we show to exhibit such a power law word rank frequency.
\end{abstract}

\section{Background and Context}
A recurring task in the context of parsing and neural sequence to sequence models---such as machine translation \cite{Sutskever:2011:GTR:3104482.3104610,NIPS2014_5346}, natural language processing \cite{Schmidhuber2014} and generative models \cite{Graves2013}---is to find an optimal path of tokens (e.g.\ words or letters) from a sequential list of probability distributions.
Such a distribution can for instance be produced at the output layer of a recurrent neural network, e.g.\ a long short-term memory (LSTM).
The goal is to decode these distributions by scoring all viable output sequences (paths) under some language model, and finding the path with the highest score.

Nowadays, the de-facto standard solution is to use a variant of beam search \cite{steinbiss1994improvements,Vijayakumar2016a,Wiseman2016,Kulikov2018,FBwav2letter} to traverse the list of all possible output strings. Beam search stores and explores a constant sized list of possible decoded hypotheses at each step, compared to a greedy algorithm that only considers the top element at each step.
Beam search thus interpolates between a simple greedy algorithm, and best-first search; but just like greedy search, beam search is not guaranteed to find a global optimum. Furthermore, beam search suffers from sensitivity to the predicted sequence length. Improving the algorithm itself \cite{Murray2018CorrectingLB,Yang2018}, as well as finding new decoding strategies \cite{Fan2018,Holtzman2019}, is an ongoing field of research.

A related task is found in transition based parsing of formal languages, such as context-free grammars \cite{Hopcroft2001,zhang2008tale,zhang2011transition,Zhu2015a,Dyer2015}.
In this model, an input string is processed token by token, and a heuristic prediction (which can be based on various types of classifiers, such as feed forward networks) is made on how to apply a transition at any one point.
As in generative models and decoding tasks, heuristic parsing employs beam search, where a constant sized list of possible parse trees is retained in memory at any point in time, and at the end the hypothesis optimising a suitable objective function is chosen. Improvements of beam search-based parsing strategies are an active field of research \cite{buckman2016transition,Bohnet2016,vilares2018transition}.

In essence, the problem of decoding a probabilistic sequence with a language model---or probabilistically parsing a formal grammar---becomes one of searching for paths in an exponentially-growing tree: since at each step or node the list of possible sequence hypotheses branches, with maximum degree equal to the number of predictions for the next tokens.
The goal is to find a path through this search space with the highest overall score.
Due to runtime and memory constraints, a tradeoff has to be made which limits any guarantees on the performance of the search strategy.

Quantum computing has shown promise as an emerging technology to efficiently solve some instances of difficult computing tasks in fields ranging from optimisation \cite{Gilyen2019,Montanaro2020}, linear algebra \cite{Harrow2009QuantumEquations,Berry2017}, number theory and pattern matching \cite{Montanaro2016,Montanaro2017}, language processing \cite{Aaronson2018,Wiebe2019}, machine learning \cite{McClean_2016,Bausch2018,Wang2019,Li2019SublinearQA}, to quantum simulation \cite{Lloyd1996,Babbush2018,Childs2019}.
While quantum computers are not yet robust enough to evaluate any of these applications on sample sizes large enough to claim an empirical advantage, a structured search problem such as language decoding is a prime candidate for a quantum speedup.

Although most na\"ive search problems can be sped up using Grover's search algorithm (or one of its variants, such as fixed point search or oblivious amplitude amplification), finding good applications for quantum algorithms remains challenging, and super-quadratic (i.e.\ \emph{faster than Grover}) speedups---such as Shor's for prime factorisation \cite{Shor1999}---are rare.
Recently, several exponentially-faster algorithms (such as quantum recommender systems \cite{Kerenedis2016recommender}, or dense low rank linear algebra \cite{Wossnig2018denseHHL}) have been proven to rely on a quantum random access memory model which, if classically available, can yield an exponential speedup without the need for quantum computing \cite{Tang2019}.

In this work, we develop a quantum search decoder for parsing probabilistic token sequences with a super-quadratic speedup as compared to its classical counterpart.
The algorithm can be seen as a generalisation of classical beam search, with potentially infinite beam width;
for finite beam width, the list of hypotheses is pruned only once at the very end---after all possible parsing hypotheses have been generated---instead of performing continuous pruning during decoding, resulting in higher accuracy guarantees.

We develop two variants of the decoder. The first one is for finding the most likely parsed string.
The more realistic use case is where the input sequence simply serves as \emph{advice} on where to find the top scoring parse under a secondary metric---i.e.\ where the element with the highest decoder score is \emph{not necessarily} the one with the highest probability of occurring when sampled.
In this variant the speedup becomes more pronounced the better the advice (see \cref{fig:grover-speedups}).

Our novel algorithmic contribution is to analyse a recently-developed quantum maximum finding algorithm \cite{VanApeldoorn2017} and its expected runtime when provided with a biased quantum sampler that we developed for formal grammars, under the premise that at each step the input tokens follow a power-law distribution;
for a probabilistic sequence obtained from Mozilla's \emph{DeepSpeech} (which we show satisfies the premise), the quantum search decoder is a power of $\approx 4-5$ faster than possible classically (\cref{fig:deepspeech-2}).

In the following we assume basic familiarity with the notion of quantum computation, but provide an overview for the reader in the supplementary material, Sec.~1.

\section{Main Results}
In this paper, we address the question of decoding a probabilistic sequence of words, letters, or generally tokens, obtained e.g.\ from the final softmax layer of a recurrent neural network, or given as a probabilistic list of heuristic parse transitions.
These models are essentially identical from a computational perspective. Hence, we give the following formal setup, and will speak of a decoding task, leaving implicit the two closely-related applications.

Given an alphabet $\Sigma$, we expect as input a sequence of random variables $X=(X_1,X_2,\ldots, X_n)$, each distributed as $X_i\sim \mathcal D_i^\Sigma$.
The distributions $\mathcal D_i^\Sigma$ can in principle vary for each $i$; furthermore, the $X_i$ can either be independent, or include correlations.
The input model is such that we are given this list of distributions explicitly, e.g.\ as a table of floating point numbers; for simplicity of notation we will continue to write $X_i$ for such a table.
The decoding machine $M$ is assumed to ingest the input one symbol at a time, and branch according to some factor $R$ at every step; for simplicity we will assume that $R$ is constant (e.g.\ an upper bound to the branching ratio at every step).
As noted, $M$ can for instance be a parser for a formal grammar (such as an Earley parser \cite{Earley1970}) or some other type of language model; it can either accept good input strings, or reject others that cannot be parsed. The set of configurations of $M$ that lead up to an accepted state is denoted by $\Omega$;
we assume that everything that is rejected is mapped by the decoder to some type of sink state $\omega\neq\Omega$.

While we can allow $M$ to make use of a heuristic that attempts to guess good candidates for the next decoding step, it is not difficult to see that a randomised input setting is more generic: we thus restrict our discussion to a decoder $M$ that processes a token sequence step by step, and such that
its state itself now simply becomes a sequence $(M_i)_{i\leq n}$ of random variables.
Described as a stochastic process, the $M_i$ are random variables over the set $\Omega$ of internal configurations after the automaton has ingested $X_i$, given that it has ingested $X_{i-1},\ldots, X_1$ prior to that, with a distribution $\mathcal D_i^\Omega$.
The probability of decoding a specific accepted string $x=(x_1,\ldots,x_n)$ is then given by the product of the conditional probabilities
\begin{align}\label{eq:product-of-probs}
\Pr(M_n=x):=&\ \mathcal N\Pr(X=x) \\
=&\ \frac{1}{\mathcal N}\prod_{i=1}^n\Pr(X_i=x_i|X_{j}=x_{j}, j\le i-1) \nonumber
\end{align}
where $\mathcal{N}= \sum_{x \in \Omega} \Pr(X=x)$.
In slight abuse of notation we write $M_n=x$ when we mean $M_n=y(x)$, where $y(x)$ is the configuration of the parser $M$ that was provided with some input to produce the parsed string $x$ (which is unambiguous as there is a one-to-one mapping between accepted strings and parser configurations $y(x)$). Similarly, we write $x\in\Omega$ for an accepted string/decoded path.

The obvious question is: which final accepted string of the decoder is the most likely?
This is captured in the following computational problem.
\begin{problem}[Most Likely Parse][MLP]\label{prob:mostLikelyParse}
\probleminput{
    Decoder $M$ over alphabet $\Sigma$, set of accepting configurations $\Omega$. Sequence of random variables $(X_i)_{i\le n}$ over sample space $\Sigma$.
}
\problemquestion{
    Find $\sigma = \argmax_{x\in\Omega} \Pr(M_n = x)$.
}
\end{problem}

Classically, it is clear that if we have a procedure that can sample the random variable $M_n$ efficiently, then we can find the most likely element with an expected runtime of $1/\Pr(M_n = \sigma)$, as this is the number of samples we are expected to draw to see the element once.
While such sampling algorithms might be inefficient to construct in general, we emphasize that the question of drawing samples from strings over a formal language is an active field of research, and algorithms to sample \emph{uniformly} are readily available for a large class of grammars: in linear time for regular languages \cite{Bernardi2012,Oudinet2013}, but also context-free grammars/restrictions thereof \cite{McKenzie97generatingstrings,Goldwurm2001,Hickey1983,Gore1997,Denise1996}, potentially with global word bias \cite{Reinharz2013,Lorenz2013,Denise2000,Ponty2012}.

In \cref{th:sampler,sec:sampler}, we lift such a classical uniform sampler to a quantum sampler (denoted $\op U_\mu$) with \emph{local} (instead of global) word bias, which we can use to obtain a quantum advantage when answering \prob{Most Likely Parse}.
We note that the techniques used to prove \cref{th:sampler} may well be used to obtain a (potentially faster) classical Monte Carlo procedure to sample from $M_n$.
In what follows, we will therefore keep the decoder's time complexity separate from the sampler's runtime and simply speak of the decoder's query complexity to $\op U_\mu$.

We prove the following result:
\begin{theorem}\label{th:1}
For an input sequence of $n$ random variables to a parser with sampling subroutine $\op U_\mu$, there exists a quantum search algorithm answering \prob{Most Likely Parse} with certainty, using $\pi/4\sqrt{\Pr(M_n=\sigma)}$ queries to $\op U_\mu$.
\end{theorem}
As explained, this theorem formalises the expected quadratic speedup of the runtime as compared to a classical algorithm based on sampling from $M_n$.
Given the input to the parser is power-law distributed (see \cref{def:powerlaw}), this allows us to formulate the following corollary.
\begin{corollary}
\label{cor:2}
If the $X_i\sim\PowerLaw_R(k)$, answering \prob{Most Likely Parse} requires at most $1/H_R(k)^{n/2}$ queries; where $H_R(k)=\sum_{i=1}^R i^{-k}$.
\end{corollary}

Yet {\em a priori}, it is not clear that the weight of a decoded path (e.g.\ the product of probabilities of the input tokens) also corresponds to the highest score we wish to assign to such a path.
This becomes obvious in the setting of a heuristic applied to a live translation:
while at every point in time the heuristic might be able to guess a good forward transition, it might well be that long range correlations strongly affect the likelihood of prior choices.
Research addressing these long-distance ``collocations''
indicates that LSTM models are capable of using about 200 tokens of context on average, but that they sharply distinguish nearby context ($\approx 50$ tokens) from the distant past. Furthermore, such models appear to be very sensitive to word order within the most recent context, but ignore word order in the long-range context (more than $50$ tokens away) \cite{Zhu2015a,Dbrowska2008longDistanceDependencies,khandelwal2018LongDistanceLSTM}.
Similarly, transformer-type architectures with self-attention---while outperforming LSTMs---feature a fixed-width context window; extensions thereof are an active field of research \cite{Al-Rfou2019,Dai2019,kitaev2020reformer}.

To address this setting formally, we assume there exists a scoring function $F:\Omega\longrightarrow\field R$, which assigns scores to all possible decoded paths. Without loss of generality, there will be one optimal string which we denote with $\tau = \argmax_{x\in\Omega}F(x)$.
Furthermore, we order all decoded strings $\Omega$ in some fashion, and index them with numbers $i=1,\ldots,|\Omega|$.
Within this ordering, $\tau$ can now be in different places---either because the heuristic guesses differently at each step, or because the input sequence varied a little.
We denote the probability that the marked element $\tau$ is at position $i$ with $p_i$. In essence, the position where $\tau$ is found is now a random variable itself, with probability mass $\Pr(\text{finding }\tau\text{ at index $i$})=p_i$.

For the decoder probabilities $\Pr(M_n=x)$ to serve as \emph{good advice} on where to find the highest-score element under the metric $F$, we demand that the final distribution over the states of the decoder puts high mass where the highest-scoring element often occurs; or formally that
\begin{equation}\label{eq:independent-scoring}
    \Pr( M_n = \text{string with index $i$} ) = p_i.
\end{equation}

To be precise, we define the following problem.
\begin{problem}[Highest Score Parse][HSP]
\probleminput{
    Decoder $M$ over alphabet $\Sigma$ and with state space  $\Omega$. Sequence of random variables $(X_i)_{i\le n}$ over sample space $\Sigma$. Scoring function $F:\Omega\longrightarrow\field R$.
}
\problempromise{\Cref{eq:independent-scoring}.}
\problemquestion{
    Find $\tau = \argmax_{x\in\Omega}F(x)$.
}
\end{problem}

What is the classical baseline for this problem?
As mentioned in \cite{Montanaro2011}, if $p_x$ is the probability that $x$ is the highest-scoring string, then in expectation one has to obtain $1/p_x$ samples to see $x$ at least once. Any procedure based on sampling from the underlying distribution $p_x$ thus has expected runtime $\sum_{x\in\Omega}\frac{1}{p_x}\times p_x = |\Omega|$.
In a sense this is as bad as possible; the advice gives zero gain over iterating the list item by item and finding the maximum in an unstructured fashion.
Yet provided with the same type of advice, a quantum computer can exhibit tremendous gains over unstructured search.

\begin{theorem}\label{th:2}
With the same setup as in \cref{th:1} but under the promise that the input tokens are iid with $X_i\sim \PowerLaw_{|\Sigma|}(k)$ over alphabet $\Sigma$ (\cref{def:powerlaw}), that the decoder has a branching ratio $R\le|\Sigma|$, and that we can uniformly sample from the grammar to be decoded,
there exists a quantum algorithm \textsc{QuantumSearchDecode} (\cref{alg:main}) answering \prob{Highest Score Parse} with an expected number of iterations
\begin{align*}
    \mathrm{RT}_1(R,k,n) &= \BigO\left( R^{nf(R,k)}\right),\\
    \text{where}\quad f(R,k) &= \log\left( \frac{H_R(k/2)}{H_R(k)^{1/2}} \right) \Big/ \log R,
\end{align*}
and where $H_R(k)$ is defined in \cref{cor:2}.

\begin{figure}
    \centering
    \includegraphics[width=0.9\columnwidth]{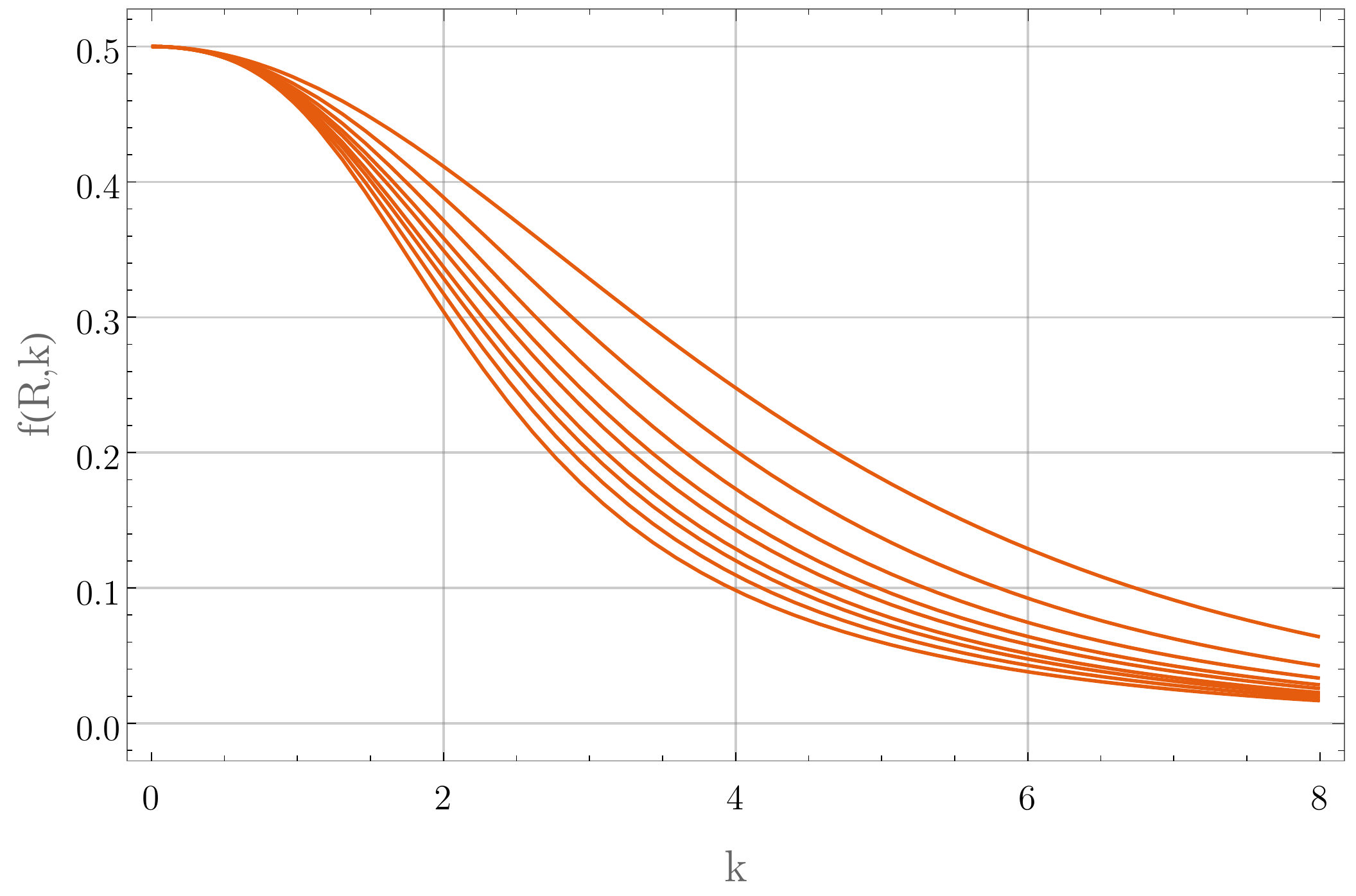}
    \caption{Exponent $f(R,k)$ of expected runtime of \textsc{QuantumSearchDecode}, when fed with a power law input with exponent $k$, over $R$ alphabet tokens; plotted are individual curves for the values $R\in\{3,5,10,15,20,30,40,60,100\}$, from top to bottom.
    For all $R$, $f(R,k)$ drops off exponentially with growing $k$.}
    \label{fig:grover-speedups}
\end{figure}

There exists no classical algorithm to solve this problem based on taking stochastic samples from the decoder $M$ that requires less than $\Omega(R^n)$ samples.
\end{theorem}

The exponent $f(R,k)$ indicates the speedup over a classical implementation of the decoding algorithm (which would have to search over $R^n$ elements).
We find that $f(R,k)<1/2$ for all $R,k>0$, and in fact $f(R,k)\longrightarrow 0$ exponentially quickly with $k$; we formulate the following corollary.
\begin{corollary}\label{cor:4}
For $k>0$, \textsc{QuantumSearchDecode} is always faster than plain Grover search (with runtime $\propto R^{n/2}$); the extent of the speedup depends on the branching ratio $R$ and the power law exponent $k$ (see \cref{fig:grover-speedups}).
\end{corollary}

Finally, in \cref{sec:beam} we modify the full quantum search decoder by only searching over the paths with likelihood above some given threshold (that we allow to depend on $n$ in some fashion), turning the decoder into a type of beam search, but where the pruning only happens at the very end (\cref{alg:beam}).
This means that in contrast to beam search, the top scoring element is found over the \emph{globally} most likely parsed paths, avoiding the risk early beam pruning brings.
We analyse the runtime of \cref{alg:beam} for various choices of beam width numerically, and analyse its performance on a concrete example---Mozilla's \emph{DeepSpeech} implementation, a speech-to-text LSTM which we show to follow a power-law token distribution at each output frame (see supplementary material, Sec.~7 for an extended discussion).

\begin{emph}
For \emph{DeepSpeech}, we empirically find that input sequence lengths of up to 500 tokens can realistically be decoded, with an effective beam width of $10^{15}$ hypotheses---while requiring $\approx 3\times 10^6$ search iterations (cf.\ \cref{fig:deepspeech-2}). 
\end{emph}
As expected, the super-Grover speedup from \cref{cor:4} is achieved in the regime where full \textsc{QuantumSearchDecoding} happens; once the beam width saturates, the speedup asymptotically approaches a quadratic advantage as compared to classical beam search.

\section{Quantum Search Decoding}\label{sec:quantum-beam-search}
In this section, we give an explicit algorithm for \textsc{QuantumSearchDecode}.
As mentioned before (see \cref{prob:mostLikelyParse}), we assume we have access to a classical sampling algorithm that, given a list of transition probabilities determined by the inputs $X_1,\ldots,X_n$, yields a random sample drawn uniformly from the distribution.
Since this sampler is given as a classical probabilistic program, we first need to translate it to a quantum algorithm.
We start with the following lemma.
\begin{lemma}\label{lem:rev} 
For a probabilistic classical circuit with runtime $T(n)$ and space requirement $S(n)$ on an input of length $n$, there exists a quantum algorithm that runs in time $\BigO(T(n)^{\log_2 3})$ and requires $\BigO(S(n)\log T(n))$ qubits.
\end{lemma}
\begin{proof}
Follows from Thm.~1 in \cite{Buhrman2001}; see supplementary material, Sec.~6.
\end{proof}

\subsection{Biased Quantum Sampling from a Regular or Context-Free Grammar}\label{sec:sampler}
Given a sampler that can yield \emph{uniformly} distributed strings $s_i$ of a language,
we want to raise it to a quantum circuit $\op U_\mu$ that produces a quantum state which is a \emph{biased} superposition over all such strings $s_i=a_{i1}a_{i2}\cdots a_{in}$, where each string is weighted by the probability $p_{ij}$ of the symbol $a_{ij}$ occurring at index $j$ (i.e.\ by \cref{eq:product-of-probs}).
In addition to the weighted superposition,  we would like to have the weight of each state in the superposition spelled out as an explicit number in an extra register (e.g.\ as a fixed precision floating point number), i.e.\ as
\begin{equation}\label{eq:mu}
    \op U_\mu\ket0 = \ket{\mu} \propto \sum_{q\in\Omega} \sqrt{p_q}\ket{h_q}\ket{p_q}\ket{q},
\end{equation}
where $\Omega$ is the set of accepted strings reachable by the decoder in $n$ steps, $\ket{h_q}$ is an ancillary state that depends on $q$ and is contained in the decoder's work space, where $q$ is a state reached by reading the input sequence $a_{q1},a_{q2},\ldots,a_{qn}$.
The weights $p_q = \prod_{j=1}^n p_{qj}$.

As outlined in the introduction, we know there exist uniform classical probabilistic samplers for large classes of grammars, e.g.\ for regular languages in linear time (e.g.\ \cite{Oudinet2013}) and polynomial time for variants of CFGs (e.g.\ \cite{Goldwurm2001}).
Keeping the uniform sampler's runtime separate from the rest of the algorithm, we can raise the sampler to a biased quantum state preparator for $\ket \mu$.
\begin{theorem}\label{th:sampler}
Given a classical probabilistic algorithm that, in time $T(n)$, produces uniform samples of length $n$ from a language,
and given a list of independent random variables $X_1,\ldots,X_n$ with pdfs $p_{ij}$ for $i=1,\ldots,n$ and $j=[\Sigma]$,
we can construct a quantum circuit $\op U_{\mu'}$ that produces a state $\ket{\mu'}$ $\epsilon$-close to the one in \cref{eq:mu}.
The algorithm runs in time $\BigO(T(n)^{1.6}\times n^3\kappa/\epsilon^2)$, where $\kappa$ is an upper bound on the relative variance of the conditional probability $\Pr(a|s_1 \dots s_i)$.
\end{theorem}
\begin{proof}
See supplementary material, Sec.~2.
\end{proof}

Getting a precise handle on $\kappa$ strongly depends on the grammar to be parsed and the input presented to it; it seems unreasonable to claim any general bounds as it will most likely be of no good use for any specific instance.
However, we note that it is conceivable that if the input is long and reasonably independent of the language to be sampled, then $\kappa$ should be independent of $n$, and  $\kappa\approx1/p(r_\mathrm{min})$, where $p(r)$ is the distribution of the input tokens at any point in time---e.g.\ $p(r)\propto r^{-k}$ as in a power law.\footnote{%
This should make intuitive sense: the branching ratios are already biased with respect to the number of future strings possible with prefix $s$; if the input sequence is independent of the grammar, then we would expect them to weigh the strings roughly uniformly; the extra factor of $1/p(r_\mathrm{min})$ simply stems from the weighing of the token we bin by, namely $a$.}

\subsection{The Quantum Search Decoder}
The quantum algorithm underlying the decoder is based on the standard maximum finding procedure developed by \cite{Durr96quantumMax,Ahuja99Max}, and its extension in \cite{VanApeldoorn2017} used in the context of SDP solvers.

The procedure takes as input a unitary operator $\op U_\mu$ which prepares the advice state, and a scoring function $F$ which scores its elements, and returns as output the element within the advice state that has the maximum score under $F$.
As in \cref{sec:sampler}, we assume that $F$ can be made into a reversible quantum circuit to be used in the comparison operation. We also note that reversible circuits for bit string comparison and arithmetic are readily available \cite{Oliveira2007comp}, and can e.g.\ be implemented using quantum adder circuits \cite{Gidney2018Adder}.

\algrenewcomment[1]{\State \textcolor{gray}{\(//\) #1}}
\begin{algorithm}[t]
\begin{algorithmic}
\Function{QuantumSearchDecode$_m$}{$\op U_\mu$, $F$}
    \State $bestScore \gets -\infty$, $counter \gets 0$
    \Repeat
        \Comment{comparator against current best score}
        \State $cmp \gets \left[ (\cdot) \mapsto (bestScore < \cdot) \right]$
        \Comment{amplify elements $\ge$ pivot}
        \State $\ket{\psi} \gets \Call{ExponentialSearch}{\op U_\mu, cmp\circ F}$
        \Comment{measure new best score}
        \State $bestScore \gets \op M_\mathrm{score}\ket\psi$
        \State $counter \gets counter + 1$
    \Until{$counter = m$}
\EndFunction
\end{algorithmic}
\caption{Algorithm for quantum search decoding.}
\label{alg:main}
\end{algorithm}

\Cref{alg:main} lists the steps in the decoding procedure.
As a subroutine within the search loop, we perform exponential search with oblivious amplitude amplification \cite{Berry14ObliviousAmpAmp}.
As in the maximum finding algorithm, the expected query count for quantum search decoding is given as follows.
\begin{theorem}
\label{thm:maximumfinding}
If $x$ is the highest-scoring string, the expected number of iterations in {\normalfont\textsc{QuantumSearchDecode}} to find the maximum is $\BigO(\min\{ 1/{\Abs{\braket{x}{\mu}}},\sqrt n\})$.
\end{theorem}
\begin{proof}
Immediate by \cite{VanApeldoorn2017}.
\end{proof}

\section{Power Law Decoder Input}\label{sec:powerlaw-input}
In this section we formally prove that if the decoder is fed independent tokens that are distributed like a power law, then the resulting distribution over the parse paths yields a super-Grover speedup---meaning the decoding speed is faster than applying Grover search, which itself is already quadratically faster than a classical search algorithm that traverses all possible paths individually.

A power law distribution is the discrete variant of a Pareto distribution, also known as Zipf's law, which ubiquitously appears in the context of language features \cite{Jager2012,Stella2016,Egghe2000,Piantadosi2014}.
This fact has already been exploited by some authors in the context of generative models \cite{Goldwater2011}.

Formally, we define it as follows.
\begin{definition}\label{def:powerlaw}
Let $A$ be a finite set with $|A|=R$, and $k>1$. Then $\PowerLaw_R(k)$ is the power law distribution over $R$ elements: for $X\sim\PowerLaw_R(k)$ the probability density function $\Pr(X=x) = r^{-k} / H_R(k)$ for an element of rank $r$, where $H_R(k)$ is the $R${\textsuperscript{th}} harmonic number of order $k$ (\cref{cor:2}).
\end{definition}
We are interested in the Cartesian product of power law random variables, i.e.\ sequences of random variables of the form $(X_1,\ldots,X_n)$.
Assuming the random variables $X_i\sim\PowerLaw_R(k)$ are all independent and of rank $r_i$ with pdf $q(r_i)=r_i^{-k}/H_R(k)$, respectively, it is clear that
\begin{equation}\label{eq:cartesian-product-of-powerlaws}
    p(r_1,\ldots,r_n) = \prod_{i=1}^n q(r_i) = \frac{1}{H_R(k)^n}\frac{1}{(r_1\cdots r_n)^k}.
\end{equation}
As in \cite{Montanaro2011}, we can upper bound the number of decoder queries in \textsc{QuantumSearchDecode} by calculating the expectation value of the iterations necessary---given by \cref{thm:maximumfinding}---with respect to the position of the top element.

We assume that at every step, when presented with choices from an alphabet $\Sigma$, the parsed grammar branches on average $R\le|\Sigma|$ times.
Of course, even within a single time frame, the subset of accepted tokens may differ depending on what the previously-accepted tokens are.
This means that if the decoder is currently on two paths $\beta_1$ (e.g.\ corresponding to ``I want'') and $\beta_2$ (``I were''), where the next accepted token sets are $\Sigma_1,\Sigma_2 \subset \Sigma$ (each different subsets of possible next letters for the two presented sentences), respectively, then we do \emph{not} necessarily have that the total probability of choices for the two paths---$\Pr(\Sigma_1)$ and $\Pr(\Sigma_2)$---are equal.
But what does this distribution over all possible paths of the language, weighted by \cref{eq:product-of-probs}, look like?

Certainly this will depend on the language and type of input presented.
Under a reasonable assumption of independence between input and decoded grammar, this becomes equivalent to answering the following question: let $X$ be a product-of-powerlaw distribution with pdf given in \cref{eq:cartesian-product-of-powerlaws}, where every term is a powerlaw over $\Sigma$.
Let $Y$ be defined as $X$, but with a \emph{random subset} of elements deleted; in particular, such that $R^n$ elements are left, for some $R < |\Sigma|$.
Is $Y$ distributed as a product-of-powerlaws as in \cref{eq:cartesian-product-of-powerlaws}, but over $R$ elements at each step? In the case of continuous variables this is a straightforward calculation (see supplementary material, Sec.~4); numerics suggest it also holds true for the discrete case.

But even if the input that the parser given is independent of the parsed grammar, it is not clear whether the \emph{sample distribution} over $R$ (i.e.\ sampling $R$ out of $|\Sigma|$ power-law distributed elements) follows the same power law as the original one over $\Sigma$;
this is in fact not the case in general \cite{Zhu2015a}.
However, it is straightforward to numerically estimate the changed power 
law exponent of a sample distribution given $R$ and $|\Sigma|$---and we note that the exponent shrinks only marginally when $R<|\Sigma|$.

In this light and to simplify the runtime analysis, we therefore assume the decoder accepts exactly $R$ tokens at all times during the parsing process (like an $R$-ary tree over hypotheses) with a resulting product-of-powerlaw distribution, and give the runtimes in terms of the branching ratio, and not in terms of the alphabet's size.
This indeed yields a fair runtime for comparison with a classical variant, since any classical algorithm will \emph{also} have the aforementioned advantage (i.e.\ we assume the size of final elements to search over is $R^n$, which precisely corresponds to the number of paths down the $R$-ary tree).

\subsection{\textsc{Most Likely Parse}: Query Bound}
In this case $F$ simply returns $p_q$ as the score in \cref{eq:mu}. 
It thus suffices to calculate the state overlap $\Abs{\braket{x}{\mu}}$, under the assumption that $x$ is the highest mass point of the probability density function.
By \cref{eq:cartesian-product-of-powerlaws}, we have $\Abs{\braket{x}{\mu}}^2 = H_R^{-n}(k)$. The claim of \cref{cor:2} follows from these observations.

\subsection{\textsc{Highest Score Parse}: Simple Query Bound}
We aim to find a top element scored under some function $F$ under the promise that $\ket\mu$ (given in \cref{eq:mu}) presents good advice on where to find it, in the sense of \cref{eq:independent-scoring}.
The expected runtimes for various power law falloffs $k$ can be obtained by taking the expectation with respect to $p_x$ as in \cite{Montanaro2011}.

In order to do so, we need to be able to calculate expecation values of the cartesian product of power law random variables, where we restrict the domain to those elements with probability above some threshold.
We start with the following observation.
\begin{lemma}\label{lem:simple-rt}
If {\textsc{QuantumSearchDecode}} receives as input iid random variables $X_1,\ldots,X_n$, with $X_i\sim\PowerLaw_R(k)$, then the number of queries required to the parser is
$
    \mathrm{RT}_1(R,k,n) = \BigO\left( H_R(k/2)^n / H_R(k)^{n/2} \right).
$
\end{lemma}
\begin{proof}
The expectation value of $1/\braket{x}{\mu}$ is straightforward to calculate; writing $\vec r=(r_1,\ldots,r_n)$, by \cref{eq:cartesian-product-of-powerlaws}, we have
\begin{align*}
    \mathds E(1/\braket{x}{\mu}) &= \sum_{\vec r}p(\vec r) \times\frac{1}{\sqrt{p(\vec r)}}   \\
    &= \frac{1}{H_R(k)^{n/2}} \sum_{r_1=1}^R\cdots\sum_{r_n=1}^R \frac{1}{(r_1\cdots r_n)^{k/2}}.
\end{align*}
As $\BigO(\min\{ 1/\braket{x}{\mu}, \sqrt n \}) \le \BigO(1/\braket{x}{\mu})$ the claim follows.
\end{proof}
We observe that the runtime in \cref{lem:simple-rt} is exponential in $n$. Nevertheless, as compared to a Grover algorithm---with runtime $R^{n/2}$---the base is now dependent on the power law's falloff $k$.
We can compare the runtimes if we rephrase $\mathrm{RT}_1(R,k,n)=R^{n f(R,k)}$, by calculating
\begin{align*}
    \left( \frac{H_R(k/2)}{H_R(k)^{1/2}} \right)^n &= R^{n f(R,k)}& \\
    \Longleftrightarrow
    f(R,k) &= \log\left( \frac{H_R(k/2)}{H_R(k)^{1/2}} \right)\Big/\log R.
\end{align*}
We observe that the exponent $f(R,k)\in(0,1/2)$, i.e.\ it is always faster than Grover, and always more than quadratically faster than classically.
The exponent's precise dependency on $k$ for a set of alphabet sizes $R$ is plotted in \cref{fig:grover-speedups}. For growing $k$, $f(R,k)$ falls off exponentially.

\subsection{\textsc{Most Likely Parse}: Full Query Bound}\label{sec:full-query-bound}

\emph{A priori}, it is unclear how much we lose in \cref{lem:simple-rt} by upper-bounding $\BigO(\min\{ 1/\braket{x}{\mu}, \sqrt n \})$ by $\BigO(1/\braket{x}{\mu})$---so let us be more precise.
In order to evaluate the expectation value of the minimum, we will break up the support of the full probability density function $p(\vec r)$ into a region where $p(\vec r)>1/R^n$, and its complement. Then, for two constants $C_1$ and $C_2$, we have for the full query complexity
\begin{align}\label{eq:full-runtime-expectation}
\mathrm{RT}_2(R,k,n) &= \mathds E\left[\BigO(\min\{ 1/\braket{x}{\mu}, \sqrt n \})\right] \\
&= C_1\!\!\!\!\!\!\! \sum_{\vec r: p(\vec r)>1/R^n} \!\!\!\!\!\! \sqrt{p(\vec r)} + C_2 \sqrt n\!\!\!\!\!\!\!\sum_{\vec r: p(\vec r)\le 1/R^n} \!\!\!\!\!\!p(\vec r). \nonumber
\end{align}
In order to calculate sums over sections of the pdf $p(\vec r)$, we first move to a truncated Pareto distribution by making the substitutions
\begin{align*}
    \sum_{r\in A} \frac{1}{r^k}
    \longrightarrow
    \int_A \frac{1}{r^k} \dd r,\ 
    H_R(k)
    \longrightarrow
    h_R(k) := \int_1^R \frac{1}{r^k}\dd r.
\end{align*}
While this does introduce a deviation, its magnitude is minor, as can be verified numerically throughout (see Fig.~1, supplementary material, where we plot both $\mathrm{RT}_1$ and the continuous variant $\mathrm{RT}_{1'}(R,k,n):=h_R^n(k/2)/h_R^{n/2}(k)$).

\DeclareDocumentCommand{\MM}{O{c} O{n} O{R} O{k_1} O{k_2}}{M_{#1,#2}^{#3,#4,#5}}
The type of integral we are interested in thus takes the form
\begin{equation}\label{eq:integral-1}
    \MM := \frac{1}{h_R^n(k_1)}\iiint_1^R\frac{\chi(r_1\cdots r_n\le c)}{(r_1\cdots r_n)^{k_2}} \dd r_1\cdots \dd r_n, 
\end{equation}
where $k_1$ is not necessarily equal to $k_2$, and typically $c=(R/h_R(k_1))^{n/k_1}$, which would reduce to the case we are seeking to address in \cref{eq:full-runtime-expectation}. Here,
$\chi(\cdot)$ denotes the characteristic function of a set, i.e.\ it takes the value 1 where the premise is true, and 0 otherwise.
We derive the following closed-form expression.
\begin{lemma}\label{lem:integral}
For $k\neq 1$, \cref{eq:integral-1} becomes
\begin{align*}
    \MM &= \frac{(-1)^n}{k'^n h_R^n(k_1)}\!\!\! \sum_{j=0}^{\min\{n,\lfloor c'/a'\rfloor\}}
    \binom nj \Bigg(\ee^{a'k'j} \\
    &\quad\quad\quad\ \ \ 
     - \ee^{-c'k'}\sum_{l=0}^{n-1}\frac{(a'k'j-c'k')^l}{l!}
    \Bigg),
\end{align*}
where $k'=1-k_2$, $c'=\log c$, $a'=\log R$.
\end{lemma}
\begin{proof}
See supplementary material, Sec.~3.
\end{proof}

\section{Quantum Beam Search Decoding}\label{sec:beam}
The goal of this section is to modify the \textsc{QuantumSearchDecoder} such that it behaves more akin to a classical beam search algorithm.
More specifically, instead of searching for the top scored element which could sit \emph{anywhere} within the advice distribution, we make the assumption that wherever the advice probability lies below some threshold $p(x)<p_0$---where $p_0$ can be very small---we discard those hypotheses.
This is done by dovetailing a few rounds of amplitude amplification to suppress all beam paths with probability less than $p_0$ (which we can do, since we have those probabilities written out as numbers within the advice state $\ket\mu$ in \cref{eq:mu}); a schematic of the algorithm can be found in \cref{alg:beam}.

Of course we only want to do this if the number of amplification rounds, given as the squareroot of the inverse of the leftover probability $\sum_{x:p(x)\ge p_0}p(x)$, is small (i.e.\ constant, or logarithmic in $n$).
We note that this expression is, as before, well-approximated by $\MM[p_0][n][R][k][k]$ given in \cref{lem:integral}.

\begin{algorithm}[t]
\begin{algorithmic}
\Function{QuantumBeamDecode$_m$}{$\op U_\mu$, $F$, $p_0$}
    \State $bestScore \gets -\infty$, $counter \gets 0$
    \Repeat
        \Comment{comparator against threshold}
        \State $cmp_1 \gets \left[ (\cdot) \mapsto (p_0 < \cdot) \right]$
        \Comment{comparator against current best score}
        \State $cmp_2 \gets \left[ (\cdot) \mapsto (bestScore < \cdot) \right]$
        \Comment{prune hypotheses}
        \State $amp \gets \left[ (\cdot) \mapsto \Call{AmplitudeAmplify}{\cdot, cmp_1} \right]$
        \Comment{select elements $\ge$ pivot}
        \State $\ket{\psi} \gets \Call{ExpoSearch}{amp \circ \op U_\mu, cmp_2\circ F}$
        \Comment{measure new best score}
        \State $bestScore \gets \op M_\mathrm{score}\ket\psi$ 
        \State $counter \gets counter + 1$
    \Until{$counter = m$}
\EndFunction
\end{algorithmic}
\caption{Algorithm for beam search decoding.}
\label{alg:beam}
\end{algorithm}

In beam search, only the top scoring hypotheses are kept around at any point in time; the difference to our method is of course that we can score the elements \emph{after every hypothesis has been built}. This is not possible in the classical case, since it would require an exponential amount of memory, or postselection.
As in \cref{sec:quantum-beam-search}, we have the two cases of finding the top scoring path and the most likely parse. Deriving a runtime bound for \textsc{Most Likely Parse} is straightforward---and does not, in fact, gain anything.
This is because when finding the maximum likelihood path $\tau$, one performs amplitude amplification on that element anyhow, and $p(\tau)>p_0$---so it is within the set of elements with probability kept intact by the post-amplification.\footnote{%
If anything, $p_0$ introduces some prior knowledge about the first pivot to pick for maximum finding.}

The only interesting case of amplifying the advice state in \textsc{QuantumSearchDecode} to raise it to a beam search variant is thus for the case of \textsc{Highest Score Parse}, using the decoder's output as advice distribution.
Instead of listing a series of results for a range of parameters, we provide an explicit example of this analysis with real-world parameters derived from Mozilla's DeepSpeech neural network in the next section, and refer the reader to Sec.~5 in the supplementary material for a more in-depth analysis of variants of a constant and non-constant amount of post-amplification.

\begin{figure}[t]
    \centering
    \includegraphics[width=10cm]{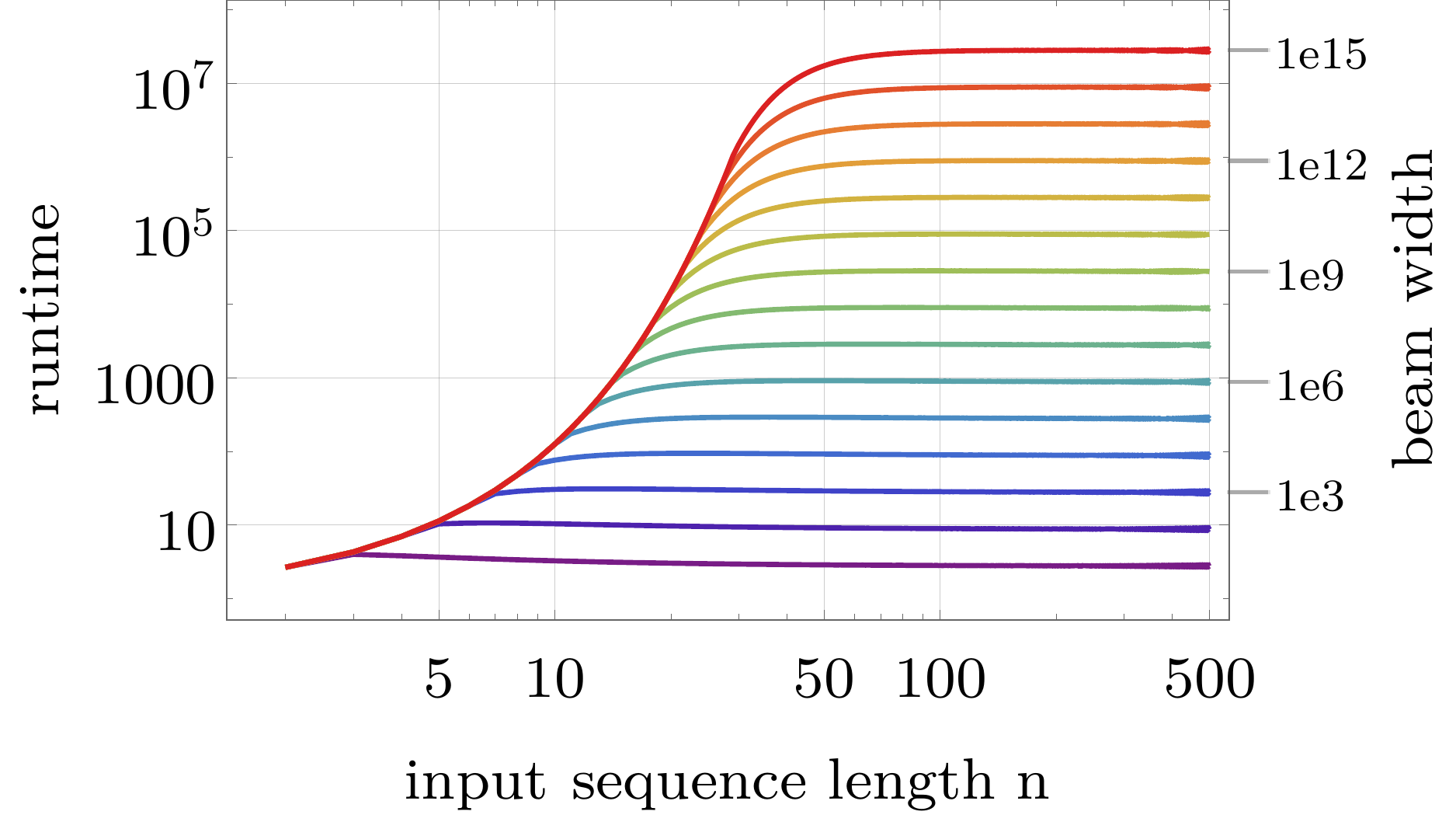}
    \caption{Runtime of quantum beam search decoding the output of Mozilla's \emph{DeepSpeech} LSTM with a grammar, assuming an average branching ratio of $R=5$, a token power law distribution with exponent $k=2.91$, and post-amplification of the quantum search decoder with a constant number of retained hypotheses $N_\mathrm{hyp}\in\{10^1,\ldots,10^{15}\}$, plotted in rainbow colors from purple to red, bottom to top.
    In the left region, where full \textsc{QuantumSearchDecoding} is performed (as the beam comprises all possible hypotheses), a super-Grover speedup is obtained (\cref{cor:4}).
    Where the beam width saturates, a Grover speedup is retained, and hypotheses are pruned only \emph{after} all hypotheses have been constructed.}
    \label{fig:deepspeech-2}
\end{figure}

\section{DeepSpeech}
\subsection{Analysis of the Output Rank Frequency}
To support the applicability of our model, we analysed our hypothesis that the output probabilities of an LSTM used to transcribe voice to letters---which can then be used e.g.\ in a dialogue system with an underlying parser---is distributed in a power-law like fashion.
More specifically, we use \emph{DeepSpeech}, Mozilla's implementation of Baidu's \emph{DeepSpeech} speech recognition system \cite{Hannun2014,Mozilla2019}; our hypothesis was that these letter probabilities follow a power-law distribution; our data supports this claim (see supplementary material, Sec.~7; also for a discussion of the LSTM's power-law output---a model feature---vs.\ the power-law nature of natural language features).


\subsection{Runtime Bounds for Quantum Beam Search Decoding}
We take the power law exponent derived from Mozilla's DeepSpeech neural network, $k=3.03$ (cf.~Sec.~5.2, supplementary material), and derive runtime bounds for decoding its output with a parser under the assumption that, on average, we take $R=5$ branches in the parsing tree at every time step.
As discussed in \cref{sec:powerlaw-input}, the sampling distribution over five elements only yields a slightly lower exponent of $k=2.91$.
How does quantum beam search perform in this setting, and how many hypotheses are actually searched over?
And what if we fix the beam's width to a constant, and increase the sequence length?
We summarise our findings in \cref{fig:deepspeech-2} (and supplementary material, Fig.~7).

\section{Summary and Conclusions}
We have presented a quantum algorithm that is modelled on and extends the capabilities of beam search decoding for sequences of random variables.
Studies of context sensitivity of language models have shown that state-of-the-art LSTM models are able to use about 200 tokens of context on average while working with standard datasets (WikiText2, Penn Treebank)  \cite{khandelwal2018LongDistanceLSTM}; state of the art transformer-based methods level off at a context window of size 512 \cite{Al-Rfou2019}.
On the other hand, under the premise of biased input tokens, our quantum search decoding method is guaranteed to find---with high constant success probability---the global optimum, and it can do so in expected runtime that is always more than quadratically faster than possible classically.
As demonstrated empirically (cf.\ \cref{fig:deepspeech-2}), our quantum beam search variant features a runtime independent of the sequence length: even for token sequences of length $>500$ the top $10^{14}$ global hypotheses can be searched for an optimal prediction, within $10^7$ steps.

We have further shown that neural networks used in the real world---concretely \emph{DeepSpeech}---indeed exhibit a strong power law distribution on their outputs, which in turn supports the premise of our algorithm.

\section{Acknowledgements}
J.\,B.\ would like to thank the Draper's Research Fellowship at Pembroke College.
S.\,S.\ would like to thank the Science Education and Research Board (SERB, Govt.\ of India) and the Cambridge Trust for supporting his PhD through a Cambridge-India Ramanujan scholarship.
We are grateful for the useful feedback and the comments we recieved from Jean Maillard, Ted Briscoe, Aram Harrow, Massimiliano Goldwurm, Mark Jerrum, and when presenting this work at IBM Zürich.
We further thank Terence Tao for the suggestion to try to take the Fourier transform of the indicator function in \cref{sec:full-query-bound}.

\linespread{1}
\printbibliography


\appendix 

\section{Quantum Computing: Preliminaries and Notation}

In this section we briefly review the basic notions and notations in quantum computation, referring to \cite{Nielsen2010} for more details. 

The usual unit of classical computation is the bit, a Boolean variable taking values in $\mathbb{Z}_2=\{0,1\}$. Its analogue in quantum computation is called the qubit, and represents the state of a physical quantum $2$-level system. A qubit can take values in $\C^2$, i.e.\ linear combinations or superpositions of two classical values (complex numbers)
\[
\alpha\ket0 + \beta\ket1
\] In particular we require that $|\alpha|^2+|\beta|^2=1$. We have also introduced the Dirac bra-ket in the above: 
\[
\ket0 := \begin{pmatrix}1\\0\end{pmatrix},~~~~~\ket1 := \begin{pmatrix}0\\1\end{pmatrix}.
\]
More generally, the set of states an $m$-qubit quantum register can take is the set of unit vectors
\begin{equation}
\label{eq:quantum-state}
\ket\phi = \sum_{i\in\{0,1\}^m}\alpha_i\ket i\text{ with } \alpha_i\in\C,\text{ such that }\sum_i|\alpha_i|^2=1
\end{equation}
in the Hilbert space spanned by a set of orthonormal basis vectors $\{\ket i, i\in\{0,1\}^m\}$, known as the \emph{computational basis}. Each $\alpha_i$ is called the \emph{amplitude of basis state $\ket i$}. We interpret the vector $\ket i$ as the $m$-dimensional complex vector $v_i$ with entries given by $(v_i)_j=\delta_{ij}$, and also interchangeably as the integer $i$ or the bit string that gives its binary representation $b_1\ldots b_m$ where $b_i$ is either $0$ or $1$. Furthermore, and of key importance to quantum mechanics and computation, the vector $\ket{b_1\ldots b_m}\in\C^{2^m}$ is interpreted as a tensor product 
\[\ket{b_1\ldots b_m}=\ket{b_1}\otimes\ket{b_2}\otimes\ldots\otimes\ket{b_m}\]
of the $m$ vectors $\ket{b_i}$ in $\C^2$. The $\otimes$ is often dropped for convenience, and we write $\ket{b_1}\ket{b_2}$ for $\ket{b_1}\otimes\ket{b_2}$.\\

\noindent\textbf{Unitary operators}: There are two ways in which we can compute on a state $\phi$. The first is by \emph{unitary evolution} of the system under the Schr\"odinger equation with a specified Hamiltonian operator $H$
\[
i\frac{d\ket\psi}{dt} = H\ket\psi,
\]
where $H$ is a hermitian matrix. Closed systems undergo reversible dynamics in quantum mechanics, and this dynamics is represented by unitary matrices. Since we can think of $\ket\phi$ as a vector in $\C^{2^m}$, a computation is represented by multiplication of this state by a $U\in\text{SU}(2^m)$, i.e. $\ket{\phi_{\text{out}}}=U\ket{\phi_{\text{in}}}$. Recall that a matrix $U$ is said to be unitary if $UU^{\dagger}=\1$, where $U^{\dagger}$ is the conjugate transpose of $U$. It is possible to compile a large `algorithm' $U$ down into elementary unitary operations, or quantum gates.\\

\noindent\textbf{Measurements}: The second kind of operation we can perform on $\ket\phi$ is measurement. For our purposes, note that the postulates of quantum mechanics say that on measuring the state $\ket\phi$ in \eqref{eq:quantum-state} \emph{in the basis $\{\ket i\}$}, we obtain as outcome the basis state $\ket i$ with probability $|\alpha_i|^2$. Since we have chosen states to be normalised, the measurement gives a valid probability mass function over the set of classical $m$-bit strings. After the measurement, the state ``collapses'' to the observed basis state $\ket{i}$, and no further information can be retrieved from the original state. \\

\noindent\textbf{Input models}: We will use two kinds of input models. The first is a quantum analogue of the classical query model, where inputs are accessed via a black-box or oracle that can be queried with an index $i$ and returns the $i$-th bit of the input bit string. For a bit string $x\in\{0,1\}^n$ we assume access to a unitary $\O_x$ which performs the map
\begin{equation}
    \O_x\ket i\ket b \ket z = \ket i \ket{b\oplus x_i} \ket{z},
\end{equation}
where the first register consists of $\lceil\log n\rceil$ qubits, the second is a single qubit register to store the output of the query, and the third is any additional workspace the quantum computer might have and is not affected by the query. Here $\oplus$ is addition on $\mathbb{Z}_2$, i.e.\ the $\cc{XOR}$ operation in Boolean logic. Note that $\O_x$ can be used by a quantum computer to make queries in superposition:
\begin{equation}
    \O_x\left(\frac1n\sum_{i=1}^n\ket i\ket b \ket z\right) = \frac1n\sum_i\ket i \ket{b\oplus x_i} \ket{z},
\end{equation}

\noindent\textbf{Complexity measures}: For many theoretical studies in complexity theory, the query input model is a powerful setting where several results have been proven. In this model, the total number of queries made to the input oracle is the primary measure of algorithmic complexity, known as the \emph{query complexity}. 

For practical purposes, it is more important to understand the the number of elementary quantum gates used to implement the unitary circuit corresponding to the algorithm in the quantum circuit model. This is known as the \emph{gate complexity} of the algorithm. The depth of the circuit is directly related to the time complexity, and gives an idea of how parallelisable the algorithm is.


\section{Biased Quantum Sampling from a Regular or Context-Free Grammar}\label{app:sampling}
In this section we rigorously proof Theorem 6, which we restate for completeness.
\setcounter{theorem}{5}
\begin{theorem}
Given a classical probabilistic algorithm that, in time $T(n)$, produces uniform samples of length $n$ from a language,
and given a list of independent random variables $X_1,\ldots,X_n$ with pdfs $p_{i,j}$ for $i=1,\ldots,n$ and $j=[\Sigma]$,
we can construct a quantum circuit $\op U_{\mu'}$ that produces a state $\ket{\mu'}$ $\epsilon$-close to the one in \cref{eq:mu}.
The algorithm runs in time $\BigO(T(n)^{1.6}\times n^3\kappa/\epsilon^2)$, where $\kappa$ is an upper bound on the relative variance of the conditional probability $Pr(a|s_1 \dots s_i)$.
\end{theorem}
\begin{proof}
Using \cref{lem:rev}, translate the parser---which takes its input step by step---into a sequence of unitaries $\op U=\op U_n\cdots\op U_1$.
Considering a single unitary $\op U_i$ at the $i^\mathrm{th}$ step, it is clear that it can be broken up into a family of unitaries $( \op U_i^a )_{a\in\Sigma}$, such that each $\op U_i^a$ is a specialization of $\op U_i$ when given a fixed input symbol $a\in\Sigma$.
We define $\op V_i^a$ to perform $\op U_i^a$, and in addition store the input $a$ in some ancillary workspace, e.g.\ via $\op V_i^a\ket{\phi}\ket{\xi}=(\op U_i^a\ket\phi)\ket{\xi \oplus a}$.
Then define the block-diagonal unitary $\op V_i := \diag(\op V_i^a)_{a\in\Sigma}$, which acts like a controlled matrix, meaning that if $\op V_i$ acts on some state $\ket{\psi}=\ket{a}\ket{\phi}$, then $\op V_i\ket{\psi} = \ket{a} \op V_i^a\ket{\phi}$.
Naturally this works in superposition as well, e.g.\ $\op V_i(\alpha\ket{a} + \beta\ket{b})\ket{\phi} = \alpha\ket{a}\op V_i^a\ket\phi + \beta\ket{b}\op V_i^b\ket\phi$.
We further assume that the $\op V_0^a$ take as initial state $\ket0\ket{q_0}$.

The final step in augmenting the parser is to extend $\op V_i$ to carry out a controlled multiplication: for a finite set of numbers $F\subset\field R$ (e.g.\ fixed precision), and $d_1,d_2\in F$, we write $\op V_i(d_1)\ket a \ket{d_2} \ket\phi = \ket a\ket{d_1 \times d_2} \op V_i^a\ket \phi$.
We denote this extended unitary for step $i$ with $\op U'_i$.

The next ingredient we take is the classical uniform language sampler. Once again using \cref{lem:rev}, we raise it to a unitary $\op W$,
which takes as input a prefix $s_m:=a_1\cdots a_m$ of the $m$ previously-seen tokens, and a list of distributions over the future weights  $W_m:=(p_{i,j})_{m<j\le n}$. These are the distribution of tokens for each of the $X_j$.
We then augment $\op W$ to a circuit $\op W'$ that quantumly performs the following classical calculations, in superposition over its input:
\begin{enumerate}
    \item Draw $S$ samples uniformly at random from the grammar starting at strings prefixed with $s_m$; denote this list with $B := \{ b_{1},\ldots,b_{S} \}$.
    \item Group the samples $B$ into bins $C_a$ of samples with the same first token $a\in\Sigma$,
    i.e.\ $C_a=\{ b \in B : b = a??\cdots? \}$, where $?$ stands for any token in the alphabet $\Sigma$.
    \item Calculate the total of the probabilities of each bin $C_a$ where each element is weighted with respect to the future probabilities given in list $W_m$, which yields a distribution $D = (d_{a})_{a\in\Sigma}$.
\end{enumerate} 
It is straightforward to write the unitary $\op W'$ that then takes a state $\ket {00}\in\mathcal H_F\otimes \field C^\Sigma$---the first register for storing a number in $F$, and the second for storing a letter---and a list of such weights $D$ to a weighted superposition $\op W'(D)\ket 0 = \sum_{a\in\Sigma} \sqrt{d_{a}}\ket{d_a}\ket{a}$ (where for the sake of simplicity we drop the scratch space register that is certainly required).
Furthermore, we need a controlled unitary $\op Q$ that, given some state $\ket{h}\ket{a}$ where $h=h(a)$ in some specified fashion---which we can demand the $\op V_i^a$ produce---uncomputes $a$ and $d_a$ from the second register, i.e.\ $\op Q\ket h\ket{d_a}\ket a=\ket h\ket{00}$.

Together with the sequence of parser unitaries $\op U_i'$, the overall quantum circuit $\op U_\mu$---depicted in \cref{fig:sampler}---can then be constructed as follows:
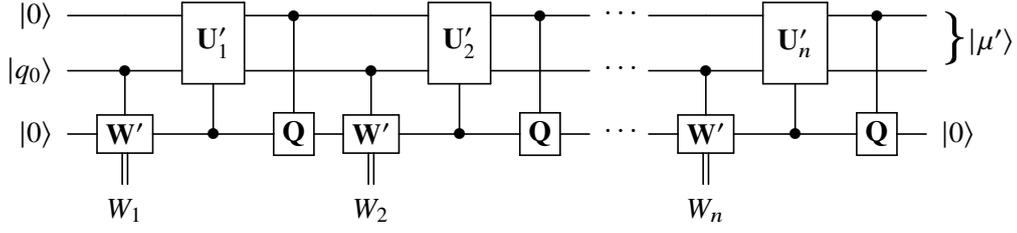
\begin{figure*}[t]
\[
    \Qcircuit @C=1em @R=1em {
    \lstick{\ket0} & \qw & \multigate{1}{\op U'_1} & \ctrl{2} & \qw & \multigate{1}{\op U'_2} & \ctrl{2} & \rstick{\cdots} \qw \\
    \lstick{\ket{q_0}}& \ctrl{1} & \ghost{\op U'_1} & \qw &\ctrl{1} & \ghost{\op U'_2} & \qw & \rstick{\cdots} \qw \\
    \lstick{\ket0}& \gate{\op W'} & \ctrl{-1} & \gate{\op Q} & \gate{\op W'} & \ctrl{-1}& \gate{\op Q} & \rstick{\cdots} \qw \\
    & \dstick{W_1}\cwx[-1] & & & \dstick{W_2}\cwx[-1] & & &
    }
    \hspace{0.78cm}
    \Qcircuit @C=1em @R=1em {
    & \qw & \multigate{1}{\op U'_n} & \ctrl{2} &\rstick{\raisebox{-1.15cm}{\scalebox{2.02}{$\}$}}\raisebox{-1.1cm}{$\ket{\mu'}$}} \qw\\
    & \ctrl{1} & \ghost{\op U'_n} & \qw &\qw \\
    & \gate{\op W'} & \ctrl{-1} & \gate{\op Q} & \rstick{\ket{0}}\qw \\
    & \dstick{W_n}\cwx[-1] & &
    }
    \vspace{.5cm}
\]
\caption{Quantum algorithm to sample from a language according to weights $W_i$, constructed in \cref{th:sampler}.}\label{fig:sampler}
\end{figure*}
For a partial string $s_1s_2\cdots s_i$ of length $i$, we denote the set of all strings in the grammar  prefixed with letters of $s$ with $\mathcal A(s_1\dots s_i)$.
At every step $i$ in the algorithm we sample the expectation value of a future hypothesis continuing with some token $a$, weighted by their individual likelihood $p_{ij}$.
The sampling procedure then yields an empirical distribution $(d_a)_{a\in\Sigma}$, which we denote with
\begin{equation}\label{eq:empirical-f}
    d_a = f^{si}_*(a) = \sum_{j=1}^S \chi\left[ b_j \in \mathcal A(s_1 \dots s_i a) \right] p(b_j) \Bigg/ \sum_{j=1}^S p(b_j),
\end{equation}
where the $S$ sampled hypothesis are given in list $B=\{ b_1,\ldots, b_S \}$ with individual letters $b_j = b_{j,1}  \cdots b_{j,n}$ ). As usual, 
$\chi[\cdot]$ denotes the indicator function, and
\[
    p(b_j) := \prod_{k=1}^n p_{k,b_{jk}}.
\]
Our goal is to show that the algorithm reproduces the desired weight distribution given in \cref{eq:mu}, i.e.
\[\Pr(s)=\prod_{i=0}^{n-1} \Pr(s_{i+1} | s_1\dots s_i) \] where \[\Pr(s_{i+1} | s_{1}\dots s_{i}) = \frac{\sum_{x \in \mathcal{A}(s_1 \dots s_i)}p(x)}{\sum_{x \in \mathcal{A}(s_{1}\dots s_{i+1})}p(x)}\]
To estimate the total probability distribution to error $\epsilon$ in total variation distance, it suffices to approximate each conditional distribution to error $\epsilon/n$, and thus we must show how many samples $S$ are required for $d_a$ to be a good estimator for $\Pr(a|s_1\dots s_i)$. 

First note that $f^{si}(a) = u^{si}(a) / v^{si}$ for
\begin{align*}
    u^{si}(a) &:= \frac1S \sum_{j=1}^S \chi\left[ b_j \in \mathcal A(s_1 \dots s_i a) \right] p(b_j) 
    \quad\text{and}\quad\\
    v^{si} &:= \frac1S \sum_{j=1}^S p(b_j)
    =\sum_{a \in \Sigma} u^{si}(a).
\end{align*}
It is straightforward to calculate that 
\begin{align*}
    \mathbb{E}(u^{si}(a))&=\frac{1}{|\mathcal{A}(s_1\dots s_i)|} \sum_{x \in \mathcal{A} (s_1\dots s_i a)} p(x) 
    \quad \text{ and } \quad\\             
    \mathbb{E}(v^{si})&=\frac{1}{|\mathcal{A}(s_1\dots s_i)|} \sum_{x \in \mathcal{A} (s_1\dots s_i)} p(x) 
\end{align*}
and so $\mathbb{E}(u^{si}(a))/\mathbb{E}(v^{si})= \Pr(a|s_1 \dots s_i)$, the value we are trying to estimate.

Therefore it suffices to take enough samples $S$ such that the $u^{si}(a)$ are close to their mean in relative error (and thus $v^{si}$ is also close in relative error, since $v^{si}=\sum_a u^{si}(a)$).

Noting that $u^{si}(a)=\frac{1}{S}\sum_{j=1}^S Y_j$ for i.i.d. random variables $Y_j$, we have that $\Var(u^{si}(a)) = \frac{1}{S} \Var(Y)$. Therefore by Chebyshev's inequality, to get a $\epsilon/n$ relative error approximation requires the number of samples $S$ to be at least 
\[S \geq \frac{\Var(Y)}{\mathbb{E}(Y)^2(\epsilon/n)^2}.\]

By assumption $\Var(Y)/\mathbb{E}(Y)^2 \leq \kappa$, and so the total number of uses of the sampler over all $n$ steps of the algorithm is $O(\kappa n^3/\epsilon^2)$ as claimed.
\end{proof}

We note that variants of this sampling algorithm are certainly possible: a na\"ive approach would be to just sample from the product-of-powerlaws distribution and postselect on the resulting strings being in the grammar; the performance of this will then depend on the number of strings in the grammar vs.\ the number of all possible strings.
Another method could be to execute the uniform sampler in superposition, and perform amplitude amplification on the resulting quantum state to reintroduce the power-law bias. The number of amplification rounds will again depend on the distribution of the strings in the grammar.

\section{\textsc{Most Likely Parse}: Full Query Bound}\label{app:full-query-bound}
In this section we rigorously prove the integral runtime expression in Lemma 10.
\begin{figure}
    \hspace{-2cm}
    \begin{minipage}{18cm}
    \includegraphics[width=5.5cm]{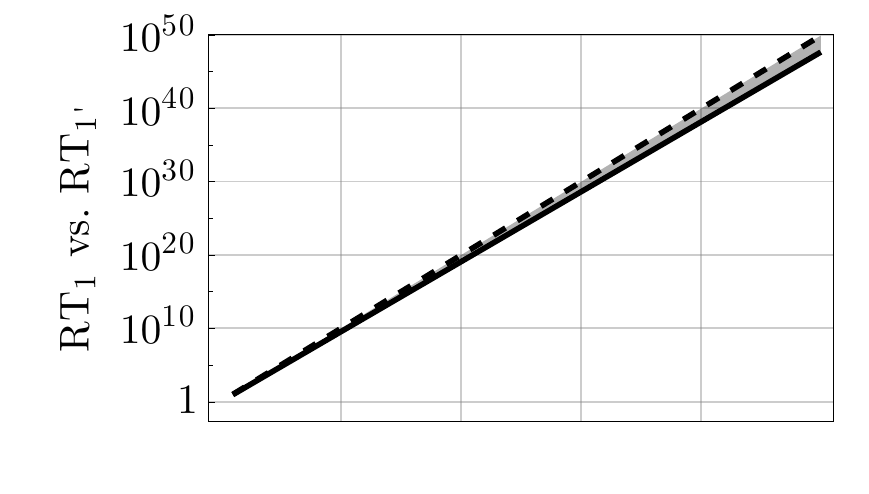}
    \nolinebreak\hspace*{-1.5cm}\nolinebreak
    \includegraphics[width=5.5cm]{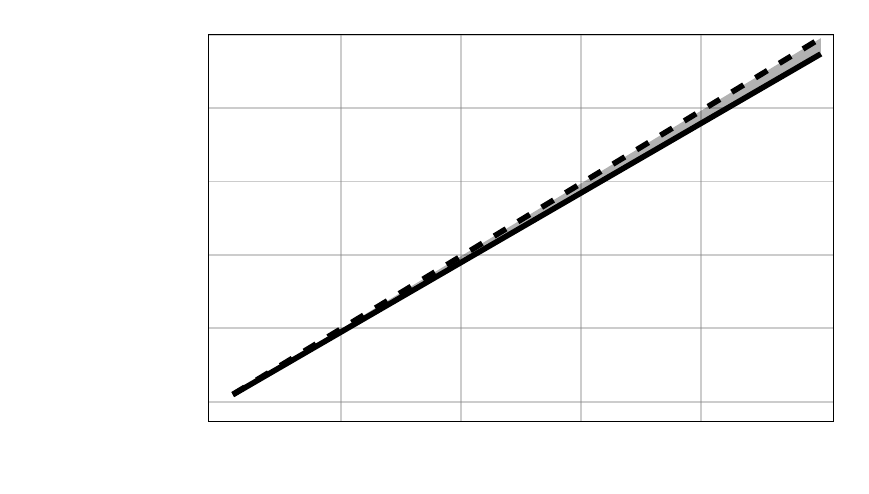}
    \nolinebreak\hspace*{-1.5cm}\nolinebreak
    \includegraphics[width=5.5cm]{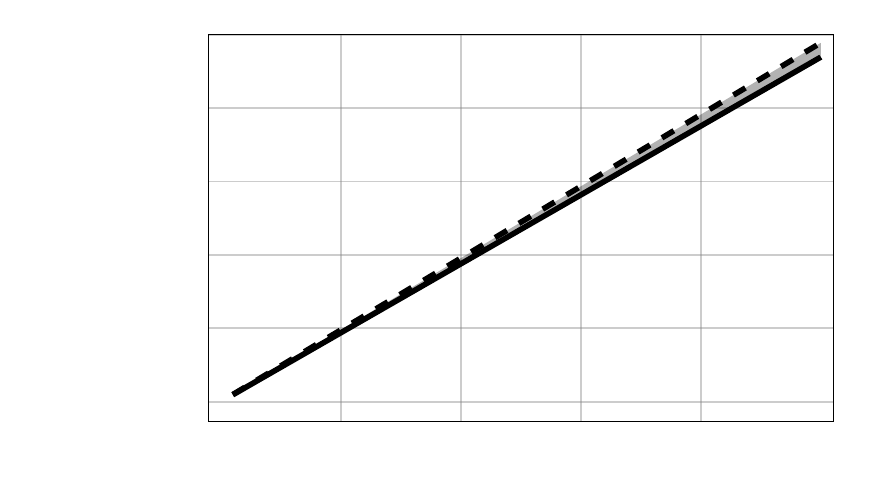}
    \nolinebreak\hspace*{-1.5cm}\nolinebreak
    \includegraphics[width=5.5cm]{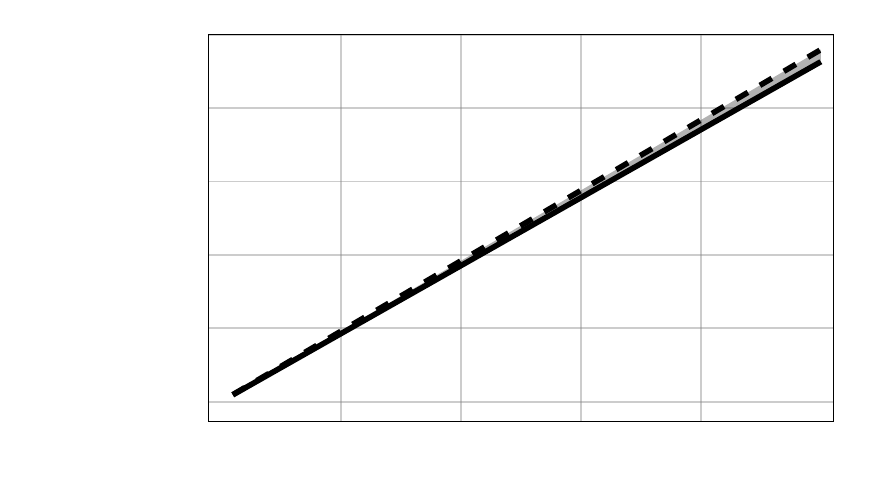}\\[-5mm]
    \includegraphics[width=5.5cm]{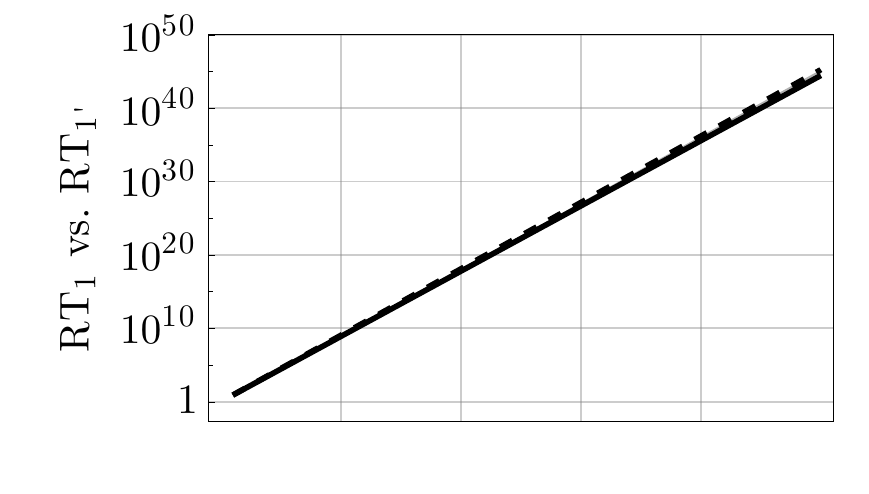}
    \nolinebreak\hspace*{-1.5cm}\nolinebreak
    \includegraphics[width=5.5cm]{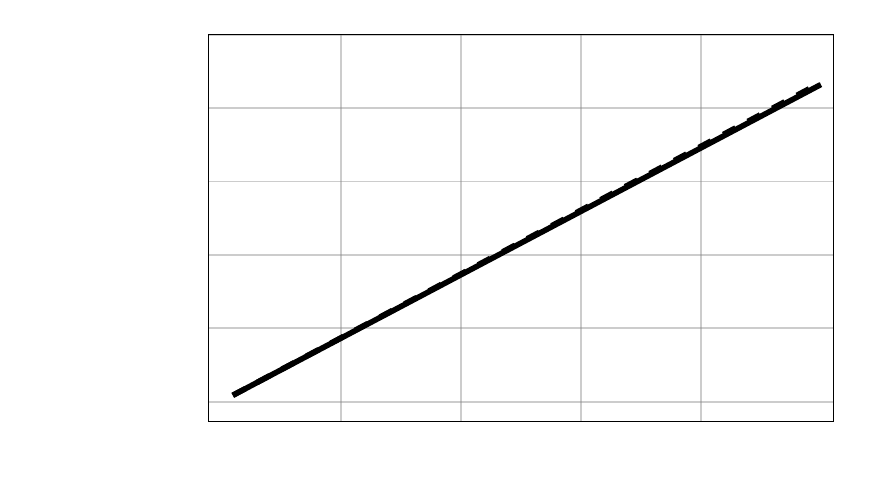}
    \nolinebreak\hspace*{-1.5cm}\nolinebreak
    \includegraphics[width=5.5cm]{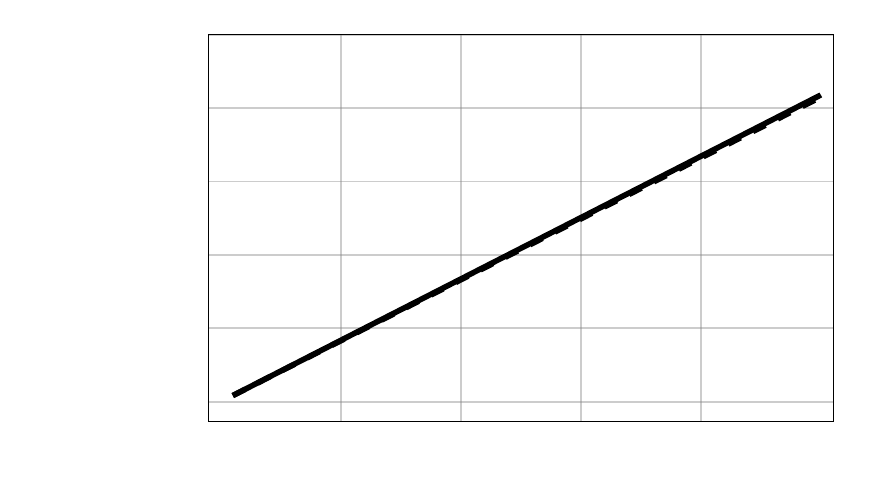}
    \nolinebreak\hspace*{-1.5cm}\nolinebreak
    \includegraphics[width=5.5cm]{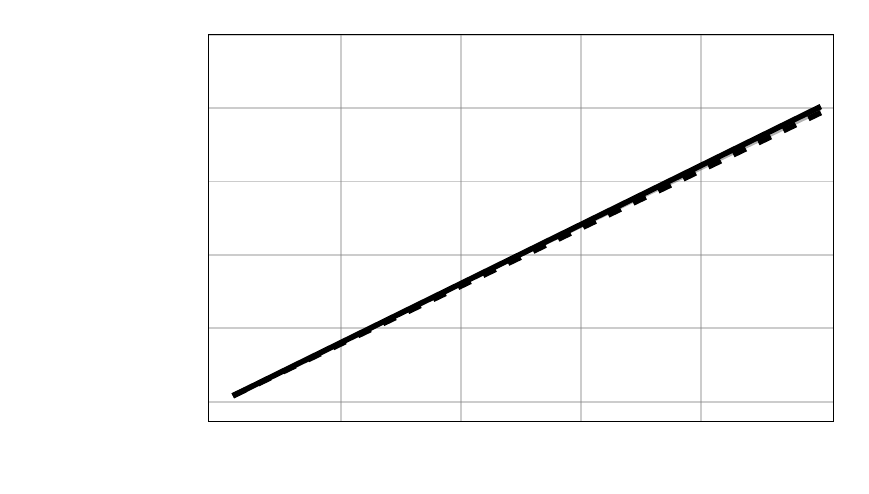}\\[-5mm]
    \includegraphics[width=5.5cm]{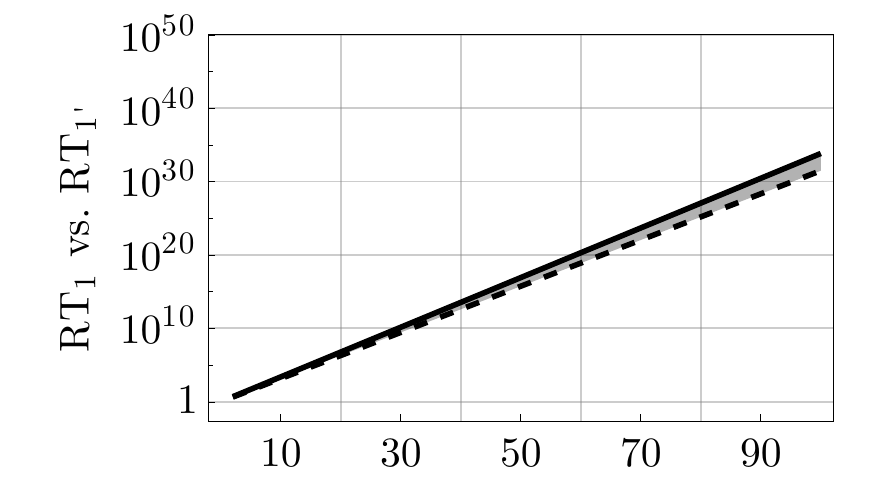}
    \nolinebreak\hspace*{-1.5cm}\nolinebreak
    \includegraphics[width=5.5cm]{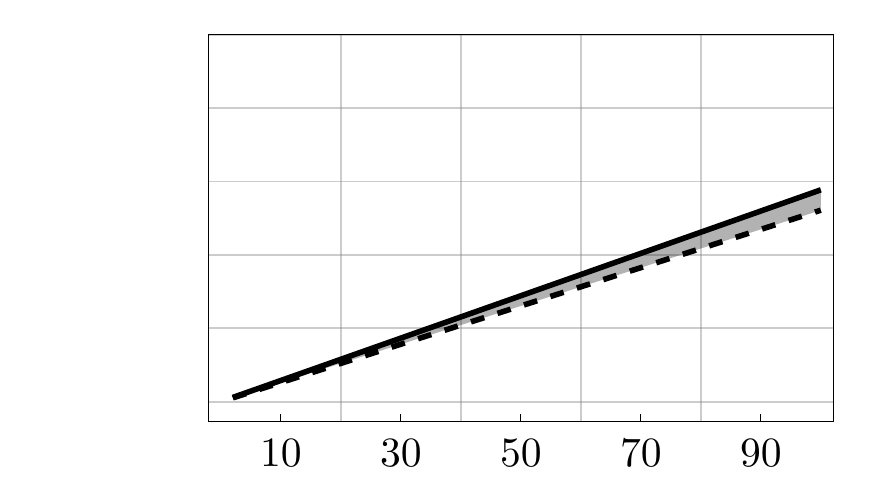}
    \nolinebreak\hspace*{-1.5cm}\nolinebreak
    \includegraphics[width=5.5cm]{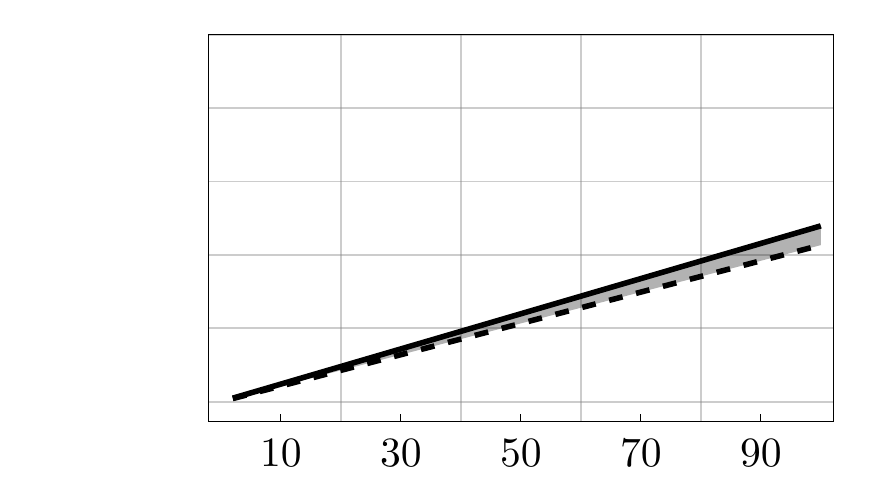}
    \nolinebreak\hspace*{-1.5cm}\nolinebreak
    \includegraphics[width=5.5cm]{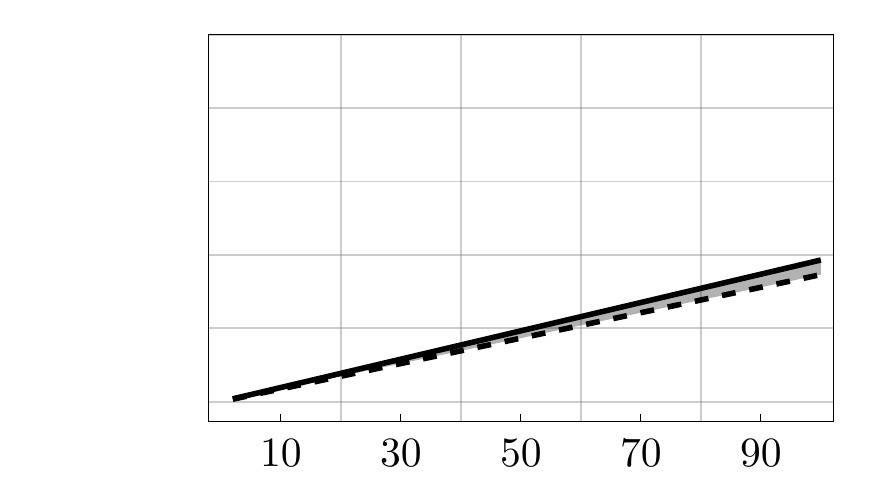}
    \end{minipage}
    \caption{Expected runtime $\mathrm{RT_1}(R,k,n)$ as evaluated for $R=10$ and various $k$ (top row: $k\in\{0.2, 0.4, 0.6, 0.8\}$, middle row: $k\in\{1.2,1.4,1.6,1.8\}$, bottom row: $k\in\{2.5,3,3.5,4\}$, always from left to right), vs.\ the same parameters used for $\mathrm{RT}_{1'}(R,k,n)$ (dashed line), where the discrete probabilities from the power law are approximated with a continuous Pareto distribution. On the $x$-axis is the length of the input sequence $n$.}
    \label{fig:RT1-vs-RT1'}
\end{figure}

As a reminder, the type of integral from \cref{eq:integral-1} we are interested in takes the form
\[
    M(R,k_1,k_2,c,n) := \frac{1}{h_R^n(k_1)}\iiint_1^R\frac{\chi(r_1\cdots r_n\le c)}{(r_1\cdots r_n)^{k_2}} \dd r_1\cdots \dd r_n, 
\]
where $k_1$ is not necessarily equal to $k_2$, and typically $c=(R/h_R(k_1))^{n/k_1}$. Here,
$\chi(\cdot)$ denotes the characteristic function of a set, i.e.\ it takes the value 1 where the premise is true, and 0 otherwise.
It is possible to integrate \cref{eq:integral-1} numerically for small $n$; however, due to the high dimensionality and the flat tail, convergence suffers drastically already for $n>6$.
Similarly, evaluating the integral with a computer algebra system takes significant time for larger $n$ and produces ever growing expressions that are hard to handle, as the reader is welcome to verify.
To address this problem, we derive the closed-form expression from \cref{lem:integral}:
\setcounter{theorem}{9}
\begin{lemma}
For $k\neq 1$, \cref{eq:integral-1} becomes
\[
    M(R,k_1,k_2,c,n) = \frac{(-1)^n}{k'^n h_R^n(k_1)}\!\!\! \sum_{j=0}^{\min\{n,\lfloor c'/a'\rfloor\}}
    \binom nj \left(
    \ee^{a'k'j} - \ee^{-c'k'}\sum_{l=0}^{n-1}\frac{(a'k'j-c'k')^l}{l!}
    \right),
\]
where $k'=1-k_2$, $c'=\log c$, $a'=\log R$.
\end{lemma}
\begin{proof}
As a first step, we perform a log substitution $z_i = \log r_i$, $\ee^{z_i}\dd z_i = \dd r_i$ which yields
\[
    M(R,k_1,k_2,c,n)=\frac{1}{h_R^n(k_1)} \iiint_0^{\log R} \ee^{(1-k_2)(z_1+\ldots+z_n)}\chi(z_1+\ldots+z_n\le\log c)\dd z_1\cdots \dd z_n.
\]
The characteristic function is now supported on a rescaled unit simplex, and writing $\bar z:=\sum_i z_i$ we can take its Fourier transform
\begin{align*}
    \mathcal F_t\chi(\bar z\le c')
    &=\frac{1}{\sqrt{2\pi}} \int_{\field R}\chi(\bar z\le c')\ee^{\ii \bar z t}\dd\bar z \\
    &=\frac{1}{\sqrt{2\pi}} \int_{-\infty}^{c'} \ee^{\ii \bar z t}\dd\bar z \\
    &=\frac{\pi}{2}\delta(t) + \frac{\ee^{\ii c' t}}{\sqrt{2\pi}\ii t} \\
    &=: \tilde \chi_{c'}(t).
\end{align*}
We of course have $\mathcal F_{\bar z}^{-1}\mathcal F_t \chi \equiv \chi$.
Then
\begin{align}
    h_R^n(k_1) M(R,k_1,k_2,c,n)
    &= \iiint_0^{a'} \ee^{k'\bar z}\chi(\bar z \le c')\dd z_1\cdots \dd z_n \nonumber\\
    &= \iiint_0^{a'} \ee^{k'\bar z}\int_{\field R} \ee^{-\ii t \bar z} \frac{\tilde \chi_{c'}(t)}{\sqrt{2\pi}} \dd t\ \dd z_1\cdots \dd z_n \nonumber\\
    &\overset*= \frac{1}{\sqrt{2\pi}} \int_{\field R} \tilde \chi_{c'}(t) \left(\prod_{l=1}^n
        \int_0^{a'} \ee^{(k'-\ii t)z_l} \dd z_l
        \right) \dd t \nonumber\\
    &= \frac{1}{\sqrt{2\pi}} \int_{\field R} \tilde\chi_{c'}(t) \left(
        \frac{\ee^{a'(k'-\ii t)}-1}{k'-\ii t}
        \right)^{\!\!n} \dd t \nonumber\\
    &= \frac{(-1)^n}{\sqrt{2\pi}} \int_{\field R} \sum_{j=0}^n \binom nj (-1)^j \frac{1}{(k'-\ii t)^n} \ee^{a'(k'-\ii t)j} \tilde\chi_{c'}(t)\dd t \nonumber\\
    &= \frac{1}{\sqrt{2\pi}} \sum_{j=0}^n \binom nj \ee^{a'k'j}\frac{(-1)^j}{\ii^n}
    \underbrace{\int_{\field R} \frac{\ee^{-\ii a' j t}}{(t+\ii k')^n} \tilde\chi_{c'}(t) \dd t}_{=:J_n}.
    \label{eq:M1}
\end{align}
In the step marked with $*$, we applied Fubini's theorem, for which we implicitly assumed a smooth limiting argument for the step function.
To evaluate the integral $J_n$, we observe that the denominator has a root of order $n$ at
\[
    t_0 = -\ii k' = \begin{cases}
    +\ii |k'| & k>1 \Leftrightarrow k'<0 \\
    -\ii |k'| & k<1 \Leftrightarrow k'>0 \\
    0 & k=1 \Leftrightarrow k'=0.
    \end{cases}
\]
We further expand the Fourier-transformed characteristic function---and again glossing over the details of Fubini's theorem to swap the integration order---to obtain
\begin{equation}\label{eq:Jn}
J_n = \frac{1}{\sqrt{2\pi}} \int_{-\infty}^{c'} \int_{\field R} \frac{\ee^{\ii t(x-ja')}}{(t+\ii k')^n} \dd t\ \dd x.
\end{equation}
We handle the integrand's three pole cases separately.
\paragraph{$\mathbf{k>1}$.}
We have $k'<0$ and an order $n$ pole at $\ii |k'|$; the integrand $g(t):=\ee^{\ii t(x-ja')}/(t+\ii k')^n$ is holomorphic in the lower half plane.
The exponent of the exponential, $x-ja'$, assumes signs
\[
    x-ja' \begin{cases}
    >0 & x>ja'\\
    <0 & x<ja'\\
    =0 & x=ja'.
    \end{cases}
\]
In the latter case, the integral (over $t$) evaluates to zero.

In the middle case, for $t=-\ii s$ we have $\exp(\ii(-\ii)s(x-ja'))=\exp(s(x-ja'))\longrightarrow 0$ as $s\longrightarrow\infty$; by Jordan's lemma we can thus write
\[
    \int_{\field R}g(t) \dd t = \lim_{r\rightarrow\infty}\oint_{\gamma_1(r)} g(t) \dd t = 0,
\]
where $\gamma_1(r)$ contains the real interval $[-r,r]$ and a half circle connecting the end points in the lower half complex plane.

In the first case, for $t=\ii s$, we have $\exp(\ii^2 s(x-ja'))=\exp(-s(x-ja'))\longrightarrow 0$ as $s\longrightarrow\infty$; however now the corresponding upper half plane loop encircles the pole of $g(x)$.
We apply the residue theorem for a flipped path $\gamma_2(r)=-\gamma_1(r)$:
\begin{align*}
\int_{\field R}\frac{\ee^{\ii t(x-ja')}}{(t+\ii k')^n} \dd t 
&= \lim_{r\rightarrow\infty}\oint_{\gamma_2(r)} \frac{\ee^{\ii t(x-ja')}}{(t+\ii k')^n} \dd t \\
&= 2\pi\ii \mathrm{Res}_t(g,t_0) \\
&= \frac{2\pi\ii}{(n-1)!} \lim_{t\rightarrow t_0} \frac{\dd^{n-1}}{\dd t^{n-1}} ((t-t_0)^n g(t)) \\
&= \frac{2\pi\ii}{(n-1)!} \lim_{t\rightarrow t_0}(\ii(x-ja'))^{n-1}\ee^{\ii t(x-ja')} \\
&= \frac{2\pi\ii^n}{(n-1)!} (x-ja')^{n-1}\ee^{k'(x-ja')}.
\end{align*}
For the case $x-ja'<0$ we are left to perform the outer integration in \cref{eq:Jn}.
If $c'\le ja'$ we necessarily have $x\le ja'$ and $J_n=0$.
For the case $c'>ja'$ we have
\begin{align*}
    \frac{1}{\sqrt{2\pi}} \int_{-\infty}^{c'} \frac{2\pi\ii^n}{(n-1)!} (x-ja')^{n-1}\ee^{k'(x-ja')}\dd x
    &= \frac{\sqrt{2\pi}\ii^n}{(n-1)!} \int_0^{c'-ja'} y^{n-1}\ee^{k'y}\dd y \\
    &= \frac{\sqrt{2\pi}\ii^n}{(n-1)!}\frac{1}{(-k')^n}(\Gamma(n) - \Gamma(n,a'k'j-c'k')) \\
    &= \frac{\sqrt{2\pi}\ii^n}{(-k')^n}\left(1-\frac{\Gamma(n,a'k'j-c'k')}{\Gamma(n)} \right),
\end{align*}
Where $\Gamma(n,\cdot)$ is the lower incomplete gamma function.
Putting it all together, we get
\begin{align*}
    J_n 
    = \frac{1}{\sqrt{2\pi}} \int_{-\infty}^{c'} \int_{\field R} \frac{\ee^{\ii t(x-ja')}}{(t+\ii k')^n}\dd t\ \dd x
    =  \frac{\sqrt{2\pi}\ii^n}{(-k')^n}\begin{cases}
    0 & c' \le ja' \\
    1-\frac{\Gamma(n,a'k'j-c'k')}{\Gamma(n)} & \text{otherwise.}
    \end{cases}
\end{align*}
Finally, we insert the last expression back into \cref{eq:M1}, and obtain
\begin{align*}
    M(R,k_1,k_2,c,n)
    &= \frac{1}{h_R^n(k_1)} \frac{1}{\sqrt{2\pi}} \sum_{j=0}^n \binom nj \ee^{a'k'j}\frac{(-1)^j}{\ii^n} J_n \\
    &= \frac{(-1)^n}{k'^n h_R^n(k_1)}\!\!\! \sum_{j=0}^{\min\{ n, \lfloor c'/a' \rfloor \}} \binom nj \ee^{a'k'j}\left( 1-\frac{\Gamma(n,a'k'j-c'k')}{\Gamma(n)} \right) .
\end{align*}
The second term in the sum we can further simplify using the identity $\Gamma(n,x)/\Gamma(n)=\ee^{-x}\sum_{l=0}^{n-1}x^l/l!$ which holds for integer $j$, which yields
\[
    M(R,k_1,k_2,c,n) = 
    \frac{(-1)^n}{k'^n h_R^n(k_1)}\!\!\! \sum_{j=0}^{\min\{ n, \lfloor c'/a' \rfloor \}} \binom nj \left(
    \ee^{a'k'j} - \ee^{-c'k'}\sum_{l=0}^{n-1}\frac{(a'k'j-c'k')^l}{l!}
    \right).
\]

\paragraph{$\mathbf{k<1}$.}
We have $k'>0$ and the order $n$ pole of \cref{eq:Jn} lies at $-\ii|k'|$.
The integrand $g(t)=\ee^{\ii t(x-ja')}/(t+\ii k')^n$ is holomorphic in the upper half plane; and analogous to before, this time when $x-ja'>0$, we have
\[
    \int_{\field R} g(t)\dd t=\lim_{r\rightarrow \infty}\oint_{\gamma_2(r)} g(t) \dd t = 0.
\]
In the opposite case we can again apply the residue theorem and obtain
\[
    \int_{\field R} g(t)\dd t \overset*= -2\pi\ii \mathrm{Res}_t(g,t_0) = -\frac{2\pi\ii^n}{(n-1)!}(x-ja')^{n-1}\ee^{k'(x-ja')},
\]
where the negative sign in step $*$ stems from the clockwise orientation of the contour $\gamma_2$.
The outer integration in \cref{eq:Jn} is now
\begin{align*}
    J_n &= -\frac{1}{\sqrt{2\pi}} \int_{-\infty}^{c'} \frac{2\pi\ii^n}{(n-1)!} \begin{cases}
        (x-ja')^{n-1} \ee^{k'(x-ja')} & x<ja' \\
        0 & \text{otherwise}
    \end{cases}\dd x\\
    &= -\frac{\sqrt{2\pi}\ii^n}{(n-1)!} \int_{-\infty}^{\min\{c',ja'\}} (x-ja')^{n-1}\ee^{k'(x-ja')}\dd x \\
    &= -\frac{\sqrt{2\pi}\ii^n}{(n-1)!} \int_{-\infty}^{\min\{c'-ja,0\}} y^{n-1}\ee^{k'y} \dd y \\
    &= -\frac{\sqrt{2\pi}\ii^n}{(n-1)!} (-1)^{n+1} k'^{-n} \Gamma(n, -k'\min\{c'-ja',0\}) \\
    &= \frac{\sqrt{2\pi}\ii^n}{(-k')^n} \frac{\Gamma(n,\min\{a'k'j-c'k',0\})}{\Gamma(n)} \\
    &= \frac{\sqrt{2\pi}\ii^n}{(-k')^n} \begin{cases}
    1 & c' \ge ja' \\
    \frac{\Gamma(n,a'k'j-c'k')}{\Gamma(n)} & \text{otherwise.}
    \end{cases}
\end{align*}
Inserting the expression back into \cref{eq:M1} we obtain
\begin{align*}
    M(R,&k_1,k_2,c,n)
    = \frac{1}{h_R^n(k_1)} \frac{1}{\sqrt{2\pi}} \sum_{j=0}^n \binom nj \ee^{a'k'j}\frac{(-1)^j}{\ii^n} J_n \\
    &= \frac{(-1)^n}{k'^n h_R^n(k_1)} \Bigg[\
    \sum_{j=0}^{\min\{n,\lfloor c'/a' \rfloor\}} \binom nj (-1)^j\ee^{a'k'j}\ \  + \\
    &\hspace{2.73cm}\sum_{j=\lfloor c'/a'\rfloor+1}^n \sum_{l=0}^{n-1} \binom nj (-1)^j\ee^{c'k'}\frac{(a'k'j-c'k')^l}{l!}
    \ \Bigg].
\end{align*}
To reduce the last sum to the previous expression, we note that
\begin{align*}
    \sum_{j=0}^n\sum_{l=0}^{n-1} \binom nj (-1)^j \frac{(xj-y)^l}{l!}
    &= \sum_{l=0}^{n-1} \frac{x^l}{l!} \sum_{j=0}^n \binom nj (-1)^j \sum_{m=0}^l \binom lm j^m \left(-\frac yx\right)^{l-m} \\
    &= \sum_{l=0}^{n-1} \frac{x^l}{l!} \sum_{m=0}^l \binom lm \left( -\frac yx \right)^{l-m} \underbrace{\sum_{j=0}^n \binom nj (-1)^j j^m}_{=(-1)^n n! S_m^{(n)}},
\end{align*}
where $S_m^{(n)}$ is the Stirling number of the second kind, which denotes the number of ways to partition a set of size $m$ into $n$ non-empty subsets.
Since $m\le l\le n-1$, $S_m^{(n)}\equiv 0$ here,
and thus
\[
 \sum_{l=0}^{n-1} \sum_{j=\lfloor c'/a'\rfloor+1}^n \binom nj (-1)^j\frac{(a'k'j-c'k')^l}{l!}
 = -\sum_{l=0}^{n-1} \sum_{j=0}^{\min\{n,\lfloor c'/a' \rfloor\}} \binom nj (-1)^j\frac{(a'k'j-c'k')^l}{l!}.
\]
The claim follows.
\end{proof}
We leave the $k=1$ case as an exercise to the reader.

With \cref{lem:integral}, we can now evaluate the terms in \cref{eq:full-runtime-expectation} efficiently.
The first term is
\begin{align}
    \mathrm{rt} = \sqrt{h_R^n(k)}M\!\!&\left[R,k,\frac k2,\left(\frac{R}{h_R(k)}\right)^{\!\frac nk},n\right],
    \label{eq:rt}
\intertext{and the second}
    \mathrm{rt}' = 1-M\!\!&\left[R,k,k,\left(\frac{R}{h_R(k)}\right)^{\!\frac nk},n\right].
\end{align}
Of interest is whether taking this full expectation value and splitting it to fall back to Grover search whenever the probability dips below $1/R^n$ yields a significant improvement of the runtime bound.
We found this to not be the case, as \cref{fig:rt-no-better} demonstrates; while for smaller $n$ there is a significant improvement, as $n$ grows the ratio $\mathrm{rt}/\mathrm{RT}_1\longrightarrow1$ exponentially fast.

\begin{figure}
    \centering
    \includegraphics[width=.32\columnwidth]{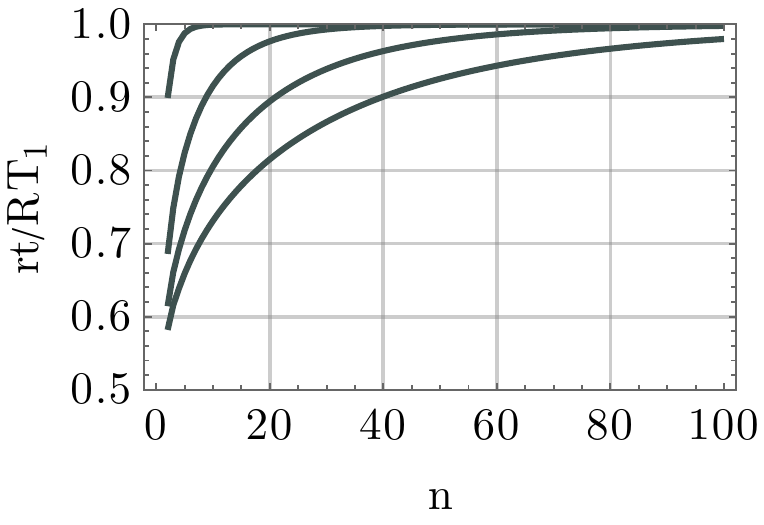}
    \hfill
    \includegraphics[width=.32\columnwidth]{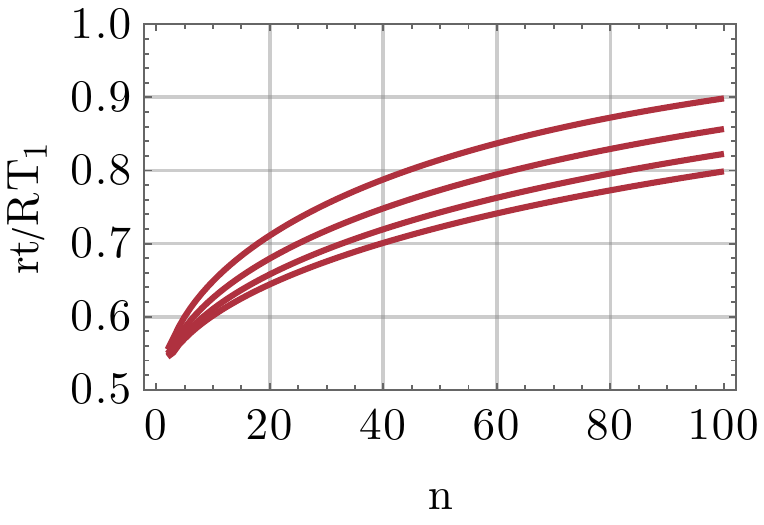}
    \hfill
    \includegraphics[width=.32\columnwidth]{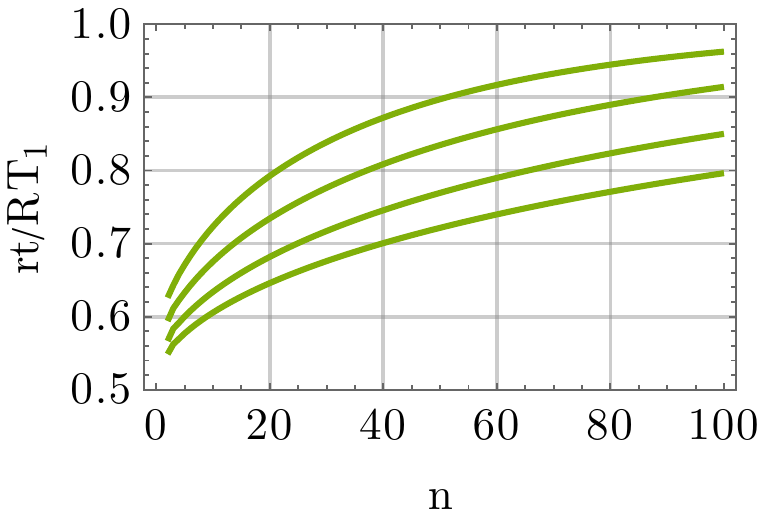}
    \caption{Ratios of $\mathrm{rt}/\mathrm{RT}_1$ for various power law exponents $k$. left: $\{0.2, 0.4, 0.6, 0.8\}$ from top to bottom, middle: $\{1.2, 1.4, 1.6, 1.8\}$ from top to bottom, right: $\{2.5, 3.0, 3.5, 4.0\}$ from bottom to top. In all cases the runtime ratos approach $1$ exponentially fast with growing $n$.}
    \label{fig:rt-no-better}
\end{figure}

\section{Postselected Product of Powerlaws}\label{sec:powerlaw-input-appdx}
In this appendix we answer the open question left in \cref{sec:powerlaw-input}.
The setup here is as follows.
Let $S>1$, and $X$ be distributed according to a product of distributions with pdf
\[
    p(r_1,\ldots,r_n) := \frac{1}{H_{S}(k)^n} \frac{1}{(r_1\cdots r_n)^k}
\]
as in \cref{eq:cartesian-product-of-powerlaws}, i.e.\ where every factor is a power law distribution over $S$ elements. If we remove a random subset of the elements such that $R^n$ (for some $R<S$) elements are left over, is the resulting probability distribution a product-of-powerlaws, where every factor is over $R$ elements?

In the continuous case this can be seen as follows.
If $X\sim\mathrm{Pareto}(k,S)$ with pdf $p$ as defined in \cref{sec:powerlaw-input}, then removing a random subset of elements on the interval $[1,S+1]$ is equivalent to taking a random characteristic function $\chi_1$ over it, with $\int_{[1,S+1]} \chi(r) \dd r=R$, and defining $X'$ with pdf $S p(r)\chi(r)/R$. We define the postselected random variable $Y$ over $[1,R]$ by relabelling the points in $\supp \chi$ by values in $[1,R]$ in an order-preserving fashion.

Similarly, if $X_n$ is a product of $n$ iid Pareto random variables with pdf $p_n$, then postselection means taking a random characteristic function $\chi_n$ on $[1,S]^n$ with 
\[ 
\int_{[1,S+1]^n} \chi_n(r_1,\ldots,r_n) \dd r_1 \cdots \dd r_n=R^n.
\]
We claim that the resulting random variable $X'_n$ with pdf $S^n p_n(r) \chi_n(r) / R^n$ then factors into a product distribution.
This holds because $\chi_n$ has the property that for all $\epsilon > 0$ there exists a bijection $f$ such that for almost all $(x_1,\ldots,x_n) \in \supp \chi_n$, there exists
\[
(y_1,\ldots,y_n)=f(x_1,\ldots,x_n) \in \left(\supp\chi_1\right)^n 
\quad\text{s.t.}\quad
\sum_{i=1}^n | x_i - y_i | < \epsilon,
\]
for some characteristic function $\chi_1$ defined on $[1,S]$. We refer to this property as $\chi_n$ being `product'.

We now prove this claim by induction on $n$. For $n=1$, X and $\chi_1$ are already product, so there is nothing to show.
Assume the hypothesis holds for $\chi_n$ which can be factored into a product $\chi_1^n$ for some $\chi_1$.
Take a random characteristic function $\chi_{n+1}$ over $n+1$ dimensions.
Let $\epsilon > 0$.
As $\field R$ is uncountable, we take a $\delta$-net over the interval $I:=[1,S+1]$ for some small $\epsilon \gg \delta>0$, which we will denote with $I_\delta$; each $x \in I$ then has a corresponding $x' \in I_\delta$ that satisfies $|x-x'| < \delta$.
In particular, $I_\delta$ is countable.
In a similar fashion, for $\chi_{n+1}$ we consider its discretized variant over $I_\delta^n$ as $\chi'_{n+1}$.

So let $(x_1,\ldots,x_n,x_{n+1}) \in \supp\chi'_{n+1}$, and analogously define the discretized characteristic functions $\chi'_n$ and $\chi'_1$.
A counting argument shows that within each $\epsilon$-bin (defined over $I$ and extended over to $I_\delta^n$ accordingly), we can map $(x_1,\ldots,x_{n+1})$ to their closest corresponding point $(y_1,\ldots,y_n,z_1)\in \chi'_n\times \chi'_1$---or if that point was previously chosen its next- and next-to-next-closest one etc., while staying within $\epsilon$ distance for each original coordinate for the majority of the points.

A limiting argument $\epsilon \longrightarrow 0$ shows that this map can be constructed for almost all points. This concludes the induction.

The question that remains is what distribution $Y$ follows. Despite scale invariance of Pareto distributions, the resulting pdf for a surviving fraction $\lambda$ of the original points looks like $\tilde p(r) = p(1+(x-1)\lambda)$, which is itself not a Pareto distribution.
Yet, since we actually work with a power law distribution, we already answered in \cref{sec:powerlaw-input} what this resulting sample distribution over $R$ looks like: it can be well-approximated by a power law with a slightly worse falloff $k'<k$ that itself can be estimated numerically in a straightforward fashion.
The smaller exponent should also account for any approximation errors made by the continuous variable analysis demonstrated in this section.

\section{Beam Search Variants}
Continuing from Sec.~5 from the main text, the relevant questions to ask here is what choice of $p_0$ will
\begin{enumerate}
    \item only require a constant---or logarithmic---number of rounds of amplitude amplification,
    \item retain a large number of hyptheses, and
    \item improve runtime for the post-amplified \textsc{QuantumSearchDecode} variant.
\end{enumerate}

We address all these questions in the next sections.

\subsection{Constant Post-Amplification}\label{sec:constant-post-amp}
In light of simplicity, we will take $\mathrm{RT}_1$ as an upper runtime bound to the full expected number of rounds, $\mathrm{RT}_2$; as we amplify away all paths with weights below the cutoff we never expect to find an element therein---meaning we can drop the fallback to Grover search in our analysis, and treat the search as if the advice state was purely on those paths with weight $\ge p_0$.

We first address the question for which choice of $p_0$ the cumulative leftover probability $M(R,k,k,p_0,n)$ can be lower-bounded by a quantity independent of $n$, which means we have to perform only a \emph{constant} number of amplitude amplification rounds on the advice state.
In order to do so, we solve the implicit inequality
\newcommand\fsplit{f_\mathrm{split}}
\begin{equation}\label{eq:fsplit}
    \text{minimize $\fsplit$ subject to\ } M\Bigg[R,k,k,\underbrace{\left(\frac{R}{h_R(k)}\right)^{\!\frac nk \fsplit}}_{=p_0},n\Bigg] \ge C_0.
\end{equation}
As $M$ is monotonically decreasing for a decreasing splitting exponent $\fsplit$, and since $M$ can be computed in $\BigO(n^2)$ many arithmetic operations, we can perform the minimization efficiently.
For a choice of $C_0=1/4$ (which implies a single amplitude amplification round) and $C_0=1/100$ (ten rounds of amplification) we plot $\fsplit$ in \cref{fig:constant-post-amp-f}.
As can be seen, $\fsplit$ tends towards a limiting value $\in(0,1)$ for $n\longrightarrow\infty$.

\begin{figure}
    \centering
    \includegraphics[width=.48\columnwidth]{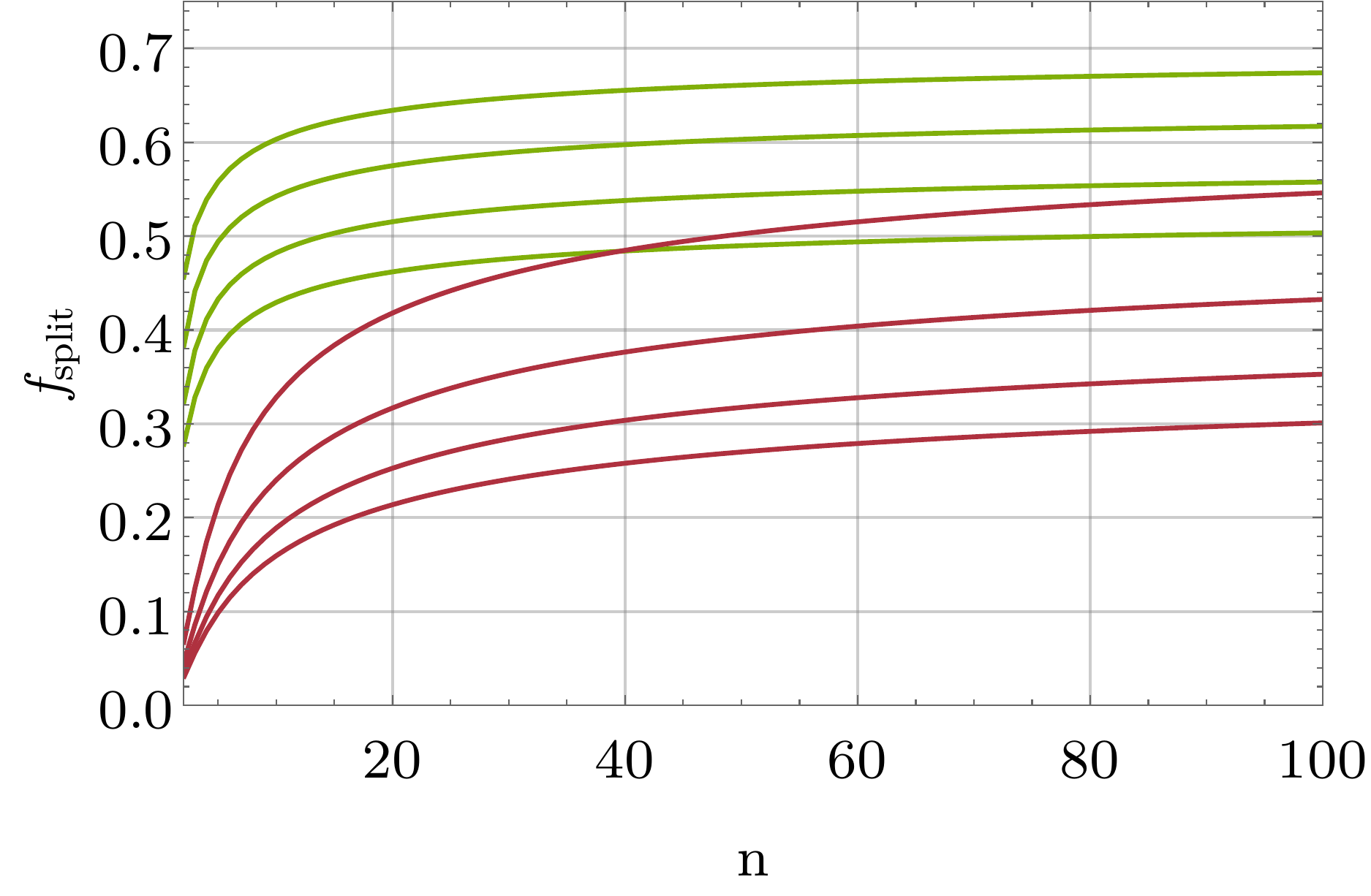}
    \hfill
    \includegraphics[width=.48\columnwidth]{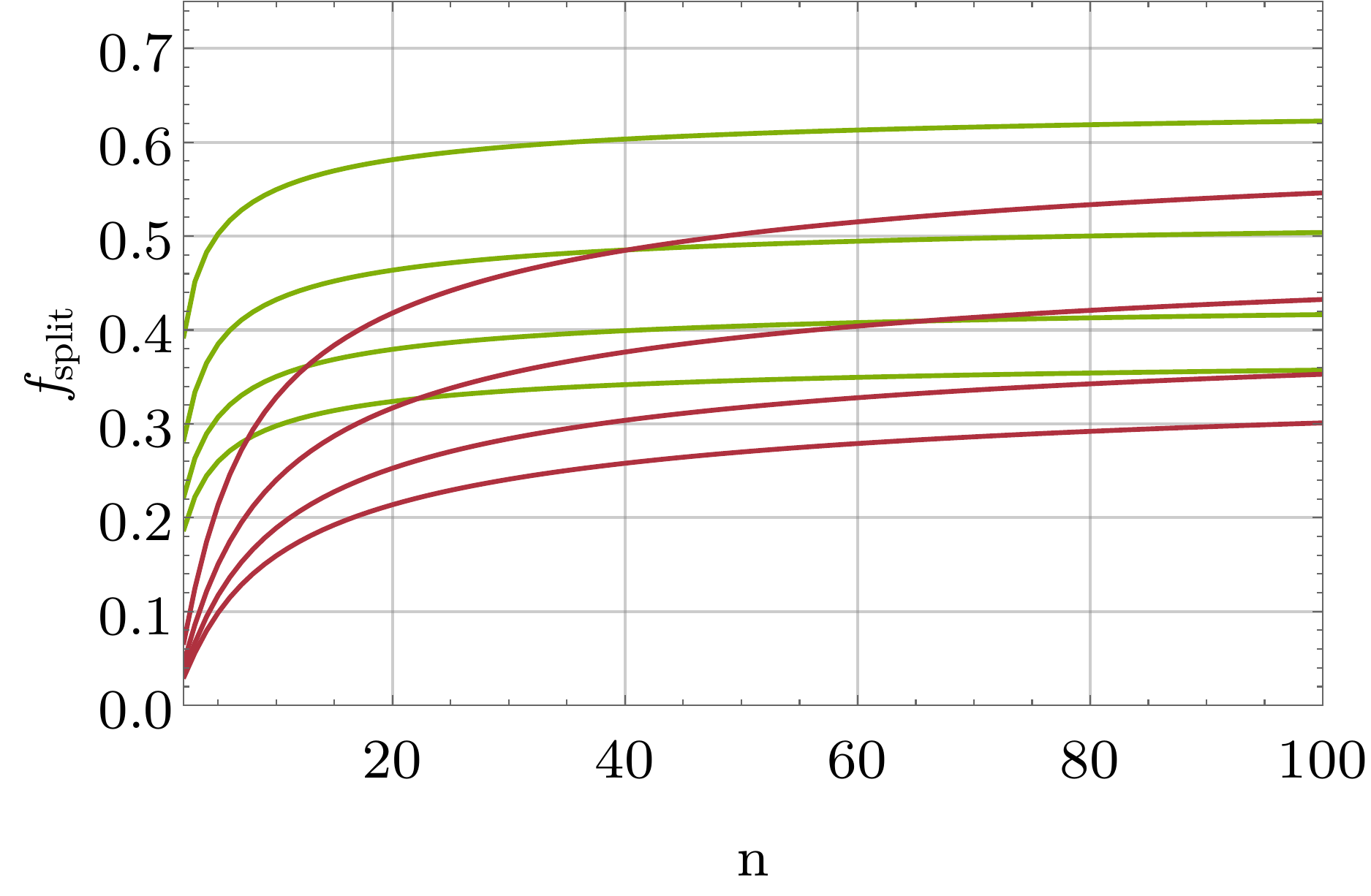}
    \caption{Minimized value of the splitting exponent $\fsplit$ as defined in \cref{eq:fsplit}.
    Plotted are the values for $R=6$ (left) and $R=24$ (right), as well as $C_0=1/4$ (green, upper family of lines) which implies exactly one extra round of amplitude amplification, and $C_0=1/100$ (red, lower family of lines) which implies ten extra rounds of amplification.
    The power law exponents chosen are $k\in\{1.5,2.0,2.5,3.0\}$ (bottom to top, respectively).}
    \label{fig:constant-post-amp-f}
\end{figure}

\newcommand\Nhyp{N_\mathrm{hyp}}
The next step in our analysis is to take the modified splitting exponent $\fsplit$ and count how many hypotheses $\Nhyp$ remain to be searched over; this is important because it is not clear a priori how many paths we can still search over, and if that quantity is low---or even tends towards zero---then we retained too few elements.
Our hope is of course that in contrast to beam search, where generally the beam's width, i.e.\ the number of hypotheses retained at any point in time, is capped at some possibly large but constant value, we have a growing number of hypotheses to search over.

In order to count this number of hypotheses given a cutoff probability $p_0$, we can evaluate $M(R,k,k,p_0,n)$ in the limit of the power law exponent $k\longrightarrow 0$, and finally multiply $h_R^n(k_1)$ in \cref{eq:integral-1} to make the integral \emph{count} instead of calculating a cumulative density.
We again choose a series of values for $R$, $k$ and $C_0$ and plot the results in \cref{fig:constant-post-amp-N}.
While the number of leftover hypotheses is indeed reduced drastically as compared to performing a full search over $R^n$ elements, it is still growing exponentially with $n$, which results in a significant number of hypotheses to search over, many more than possible in the classical setting.

\begin{figure}
    \centering
    \includegraphics[width=.48\columnwidth]{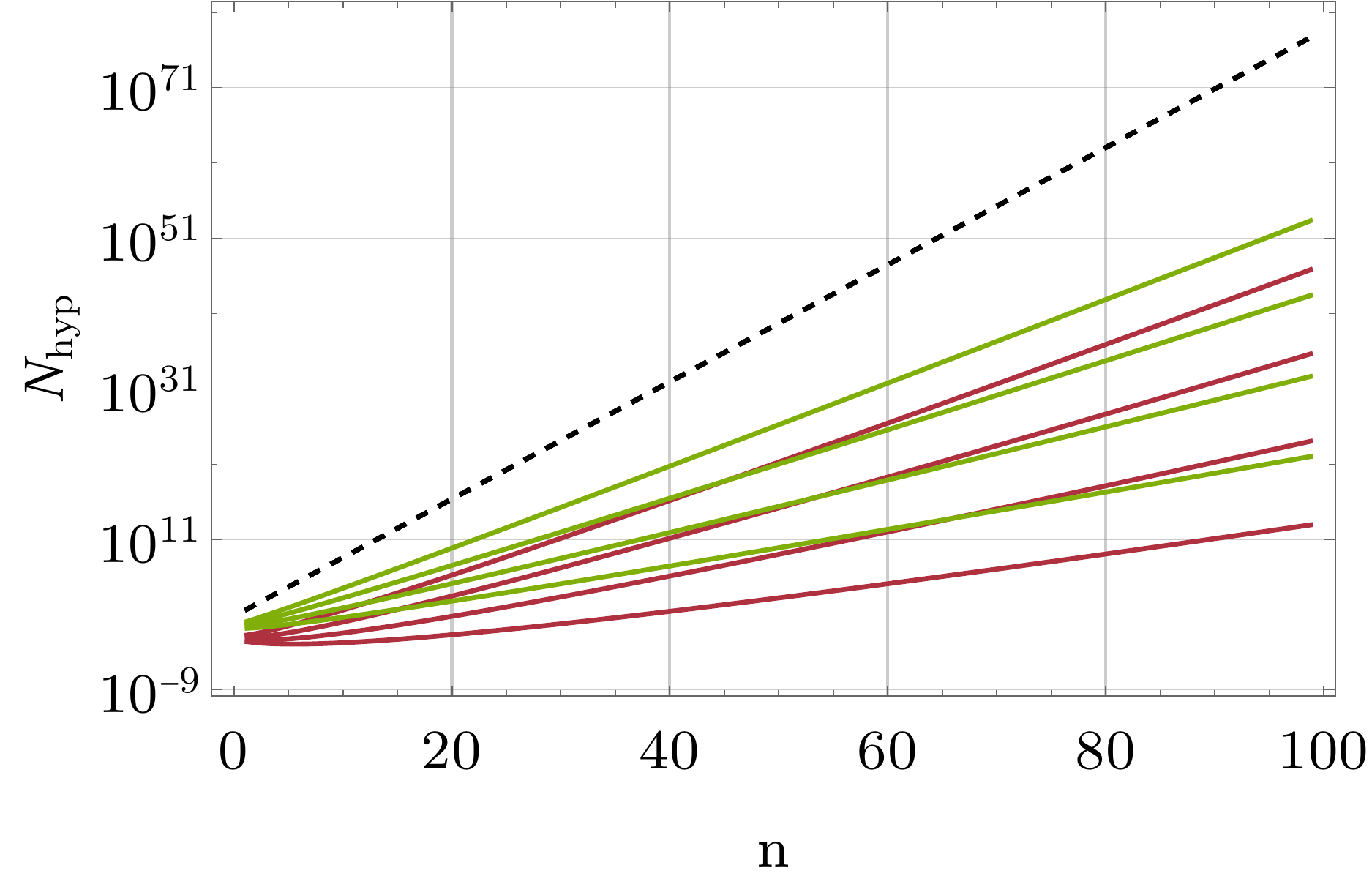}
    \hfill
    \includegraphics[width=.48\columnwidth]{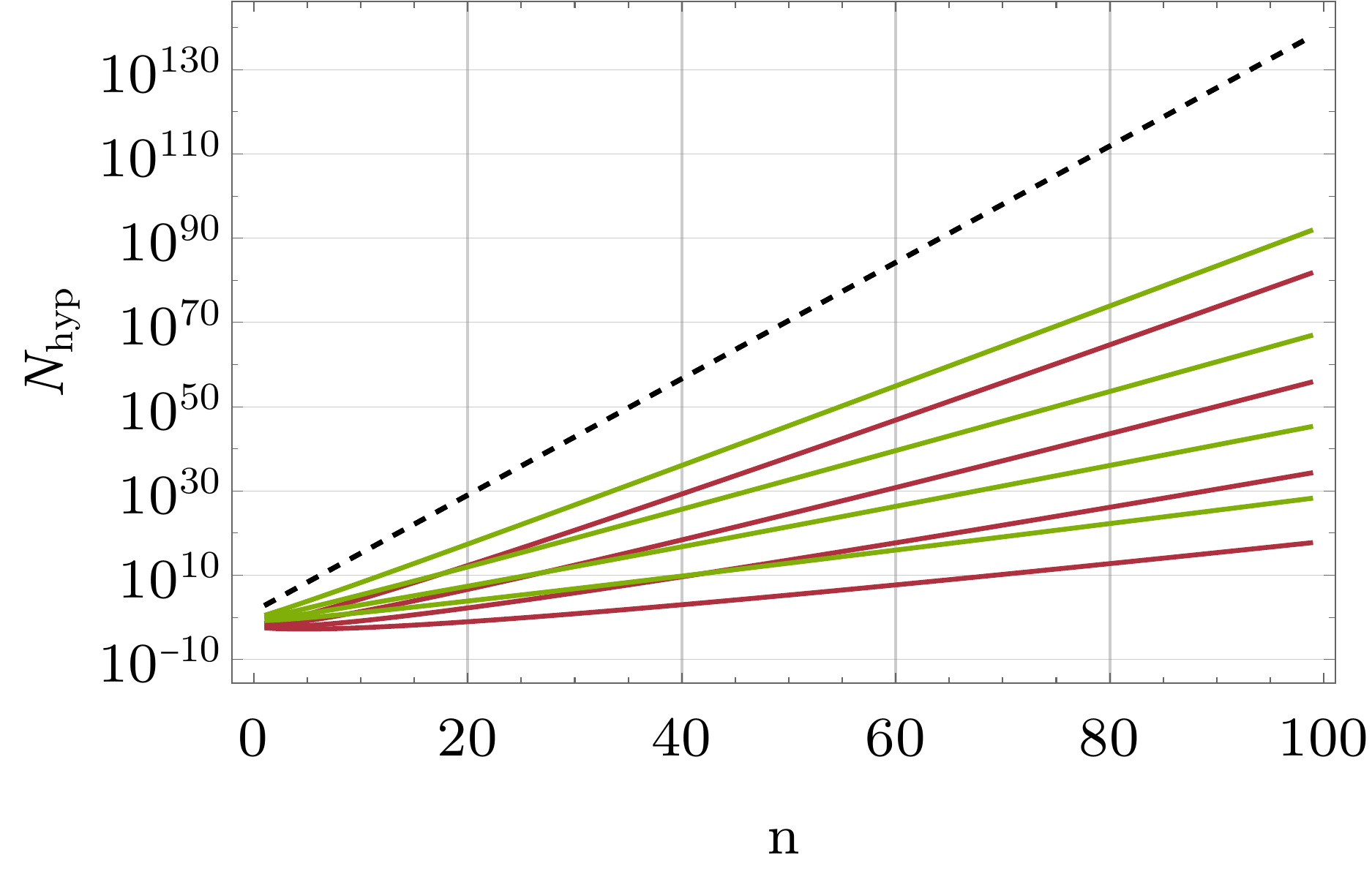}
    \caption{Number of hypotheses $\Nhyp$ left for a specific choice of splitting exponent $\fsplit$ to retain $C_0>1/4$ (green, one extra round of amplification) and $C_0>1/100$ (red, ten extra rounds of amplification) total probability weight for the hypotheses. The value of $\fsplit$ is obtained numerically from \cref{eq:fsplit} (cf.\ \cref{fig:constant-post-amp-f}).
    Plotted is the case $R=6$ (left) and $R=24$ (right), and $k\in\{1.5,2.0,2.5,3.0\}$ (from top to bottom in each plot and each color, respectively).
    The dashed line is the total number of possible hypotheses $R^n$ as reference.
    }
    \label{fig:constant-post-amp-N}
\end{figure}

As a last step, we want to analyse the modified runtime given the changed probability cutoff, which corresponds to evaluating the integral
$M(R,k,k/2,p_0,n)$ with the $p_0$ derived from the optimization \cref{eq:fsplit}.
The results are collected in \cref{fig:constant-post-amp-rt}.
As one can verify, the runtime does remain asymptotically exponential in the sequence length $n$; however the base of the exponential is reduced accordingly.

\begin{figure}
    \centering
    \includegraphics[width=.48\columnwidth]{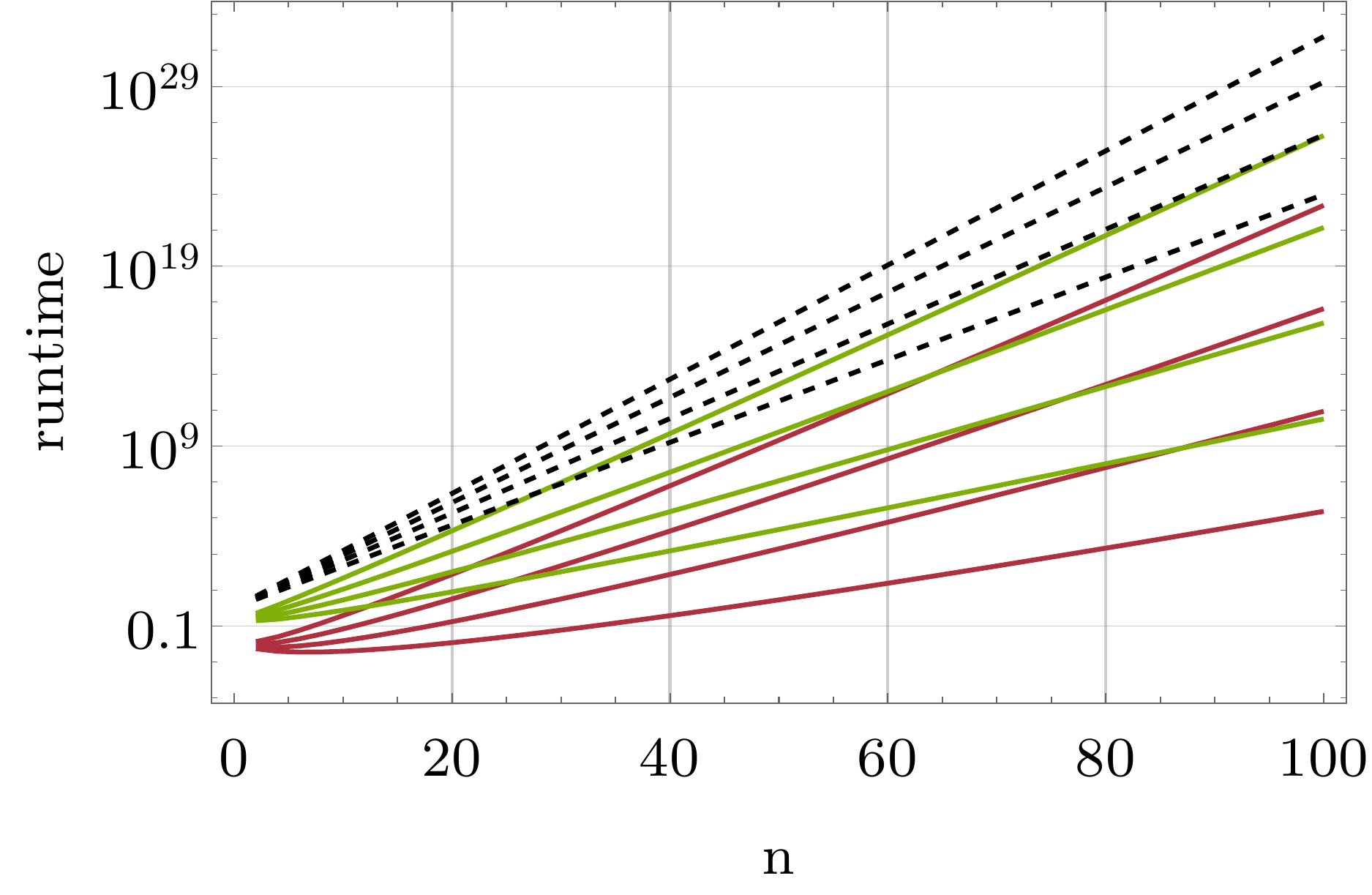}
    \hfill
    \includegraphics[width=.48\columnwidth]{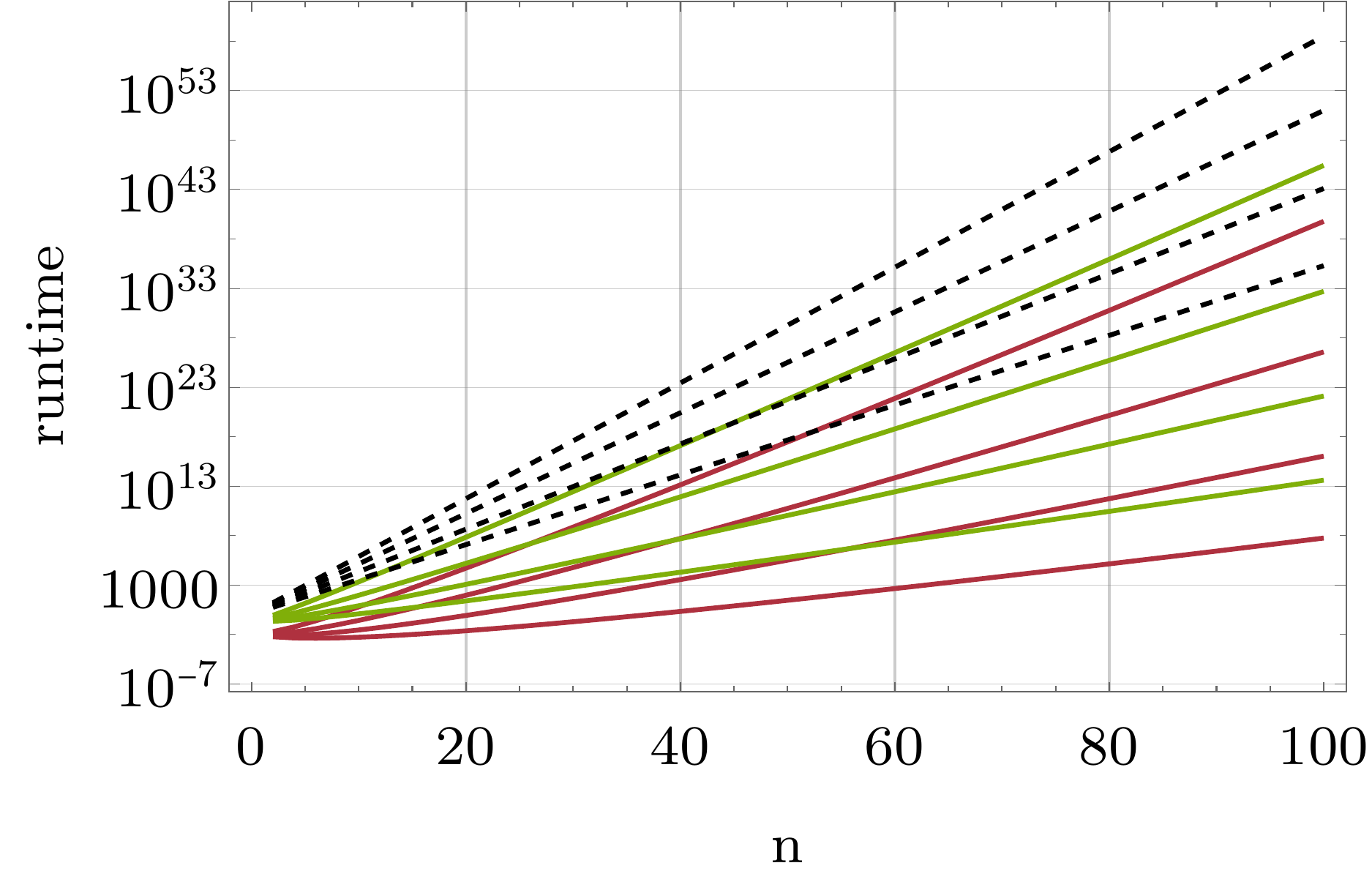}
    \caption{Runtime when post-amplifying to retain only a fraction $C_0\ge1/4$ of weight (green, one extra round of amplification) or $C_0\ge1/100$ (red, ten extra rounds of amplification) on the hypotheses. The value of $\fsplit$ is obtained numerically from \cref{eq:fsplit} (cf.\ \cref{fig:constant-post-amp-f}).
    Plotted is the case $R=6$ (left) and $R=24$ (right), and $k\in\{1.5,2.0,2.5,3.0\}$ (from top to bottom in each plot and for each color, respectively).
    The dashed line is the full search runtime $\mathrm{RT}_1(R,k,n)$ from \cref{lem:simple-rt} as reference.}
    \label{fig:constant-post-amp-rt}
\end{figure}

\subsection{Non-Constant Post-Amplification}\label{sec:non-const-amp}
The analysis of \cref{sec:constant-post-amp} can of course be repeated for a non-constant $\fsplit$; however, one has to be aware that these extra amplitude amplification rounds factor into the overall runtime.
For a retained fraction $g(n)$ of the total probability weight, the optimization thus reads
\begin{align}
    \text{minimize $p_0$ subject to\ } &M(R,k,k,p_0,n) \ge g(n) \label{eq:frac-beam}\\[2.5mm]
    \text{which retains\ } \lim_{k\rightarrow0} &M(R,k,k,p_0,n)\ \text{hypotheses,} \label{eq:N-beam}\\
    \text{and has runtime bound\ } &g(n)^{-1/2}M(R,k,k/2,p_0,n). \label{eq:rt-beam}
\end{align}

We take the power law exponent derived from Mozilla's DeepSpeech neural network, $k=3.03$ (cf.~Sec.~5.2, supplementary material), and derive runtime bounds for decoding its output with a parser under the assumption that, on average, we take $R=3$ branches in the parsing tree at every time step.
As discussed in \cref{sec:powerlaw-input}, the sampling distribution over three elements only yields a slightly lower exponent of $k=2.91$.

As an example we consider an input sequence of length 500; with the above parameters and a splitting exponential $\fsplit = n^{-1/2}$ (resp.\ $=n^{-3}$) we can search over $N_\mathrm{hyp}\approx 10^{60}$ (resp.\ $\approx 10^{18}$) hypotheses, with a runtime $\approx 10^{30}$ (resp.\ $\approx 10^9$).
Similarly, when capping the beam width at $N_\mathrm{hyp}\le 10^6$, we asymptotically require $\approx 10^3$ iterations of the beam search decoder (which includes the post-amplification rounds); for shorter sequences, a super-Grover speedup as present in full \textsc{QuantumSearchDecode} is achieved.

\begin{figure}
    \centering
    \includegraphics[width=.48\columnwidth]{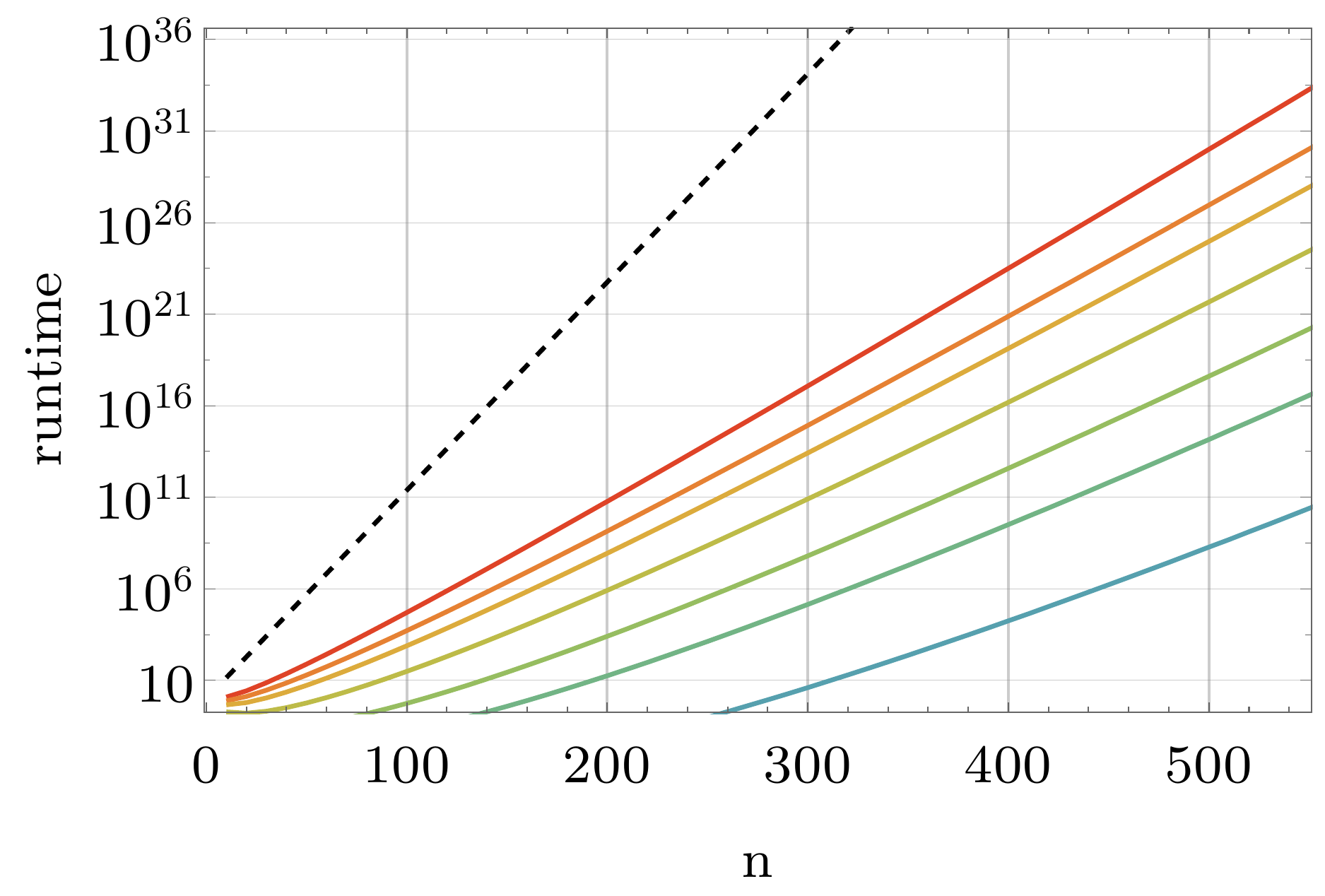}
    \hfill
    \includegraphics[width=.48\columnwidth]{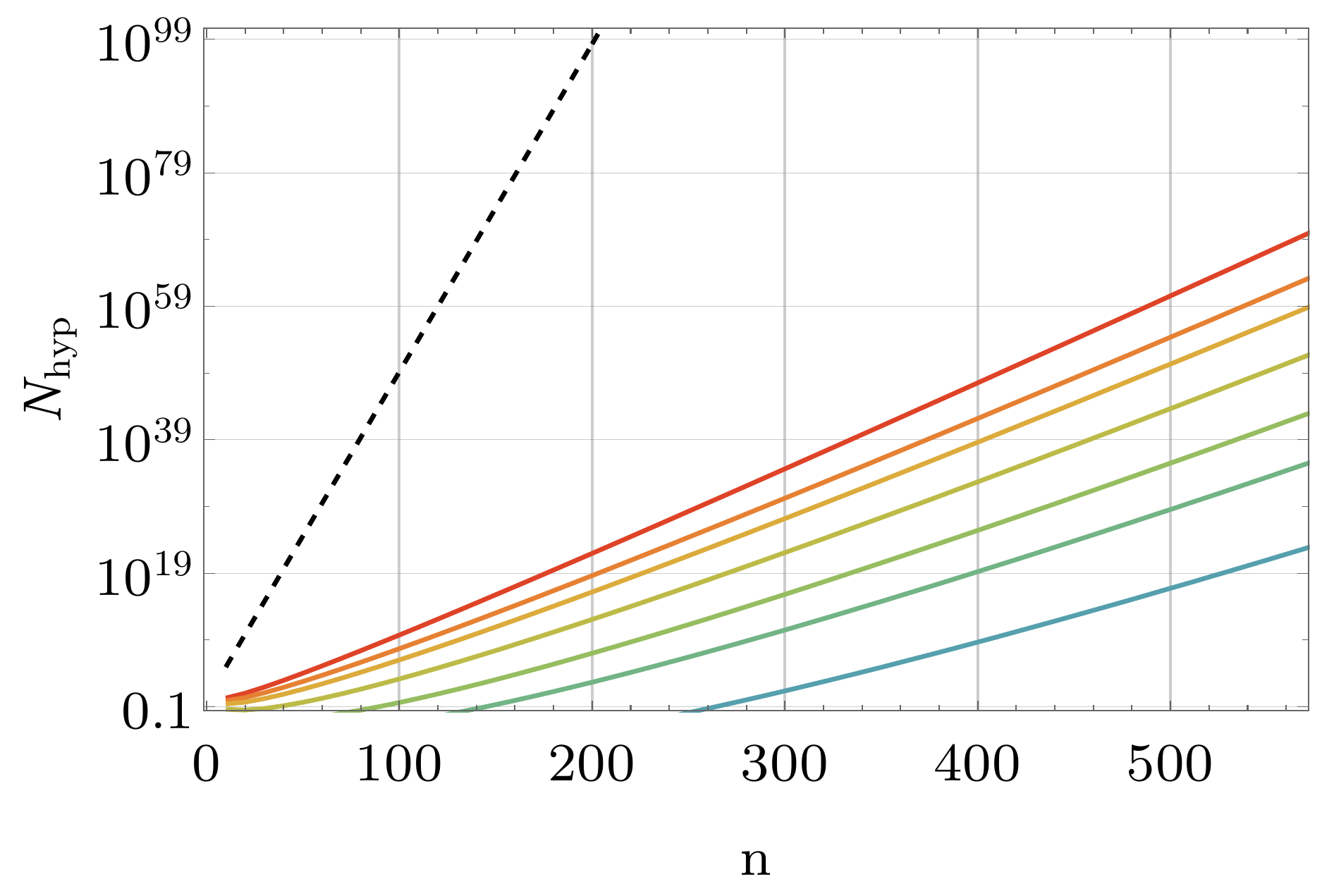}
    \caption{Number of iterations (\cref{eq:rt-beam}) and number of hypotheses (\cref{eq:N-beam}) of quantum beam search decoding the output of Mozilla's \emph{DeepSpeech} LSTM with a grammar, assuming an average branching ratio of $R=3$, a token power law distribution with exponent $k=2.91$, and post-amplification of the quantum search decoder with a retained fraction of hypotheses $C_0=C_0(n)\in\{n^{-1/2}, n^{-2/3}, n^{-1}, n^{-3/2}, n^{-2}, n^{-3}\}$  as defined in \cref{eq:frac-beam}, which is plotted in rainbow colors from red to blue, top to bottom. The dashed line is the full quantum search runtime and number of hypotheses from \cref{eq:rt}.}
    \label{fig:deepspeech}
\end{figure}

\section{Further Proof Details}
For Lemma 5, a more detailed proof is given as follows.
\begin{proof}[Lemma 5]
By Th.~1, \cite{Buhrman2001}, we have that any non-reversible computation requiring time $T$ and space $S$ can be simulated reversibly in time $T'=3^k2^{\BigO(T/2^k)}$ and space $S'=(1+\BigO(k))S$, for a $0\le k\le \log_2 T$ chosen arbitrarily.
Choose $k=\log_2 T$, then $S'=(1+\BigO(\log_2 T))S$, and $T'=\BigO(T^{\log_2 3})$.
Now translate this reversible probabilistic classical circuit into a quantum circuit---e.g.\ using the Solovay-Kitaev theorem \cite{Nielsen2010}, which incurs an at most logarithmic runtime overhead.
\end{proof}

\begin{figure}[t]
    \centering
    \includegraphics[width=0.7\columnwidth]{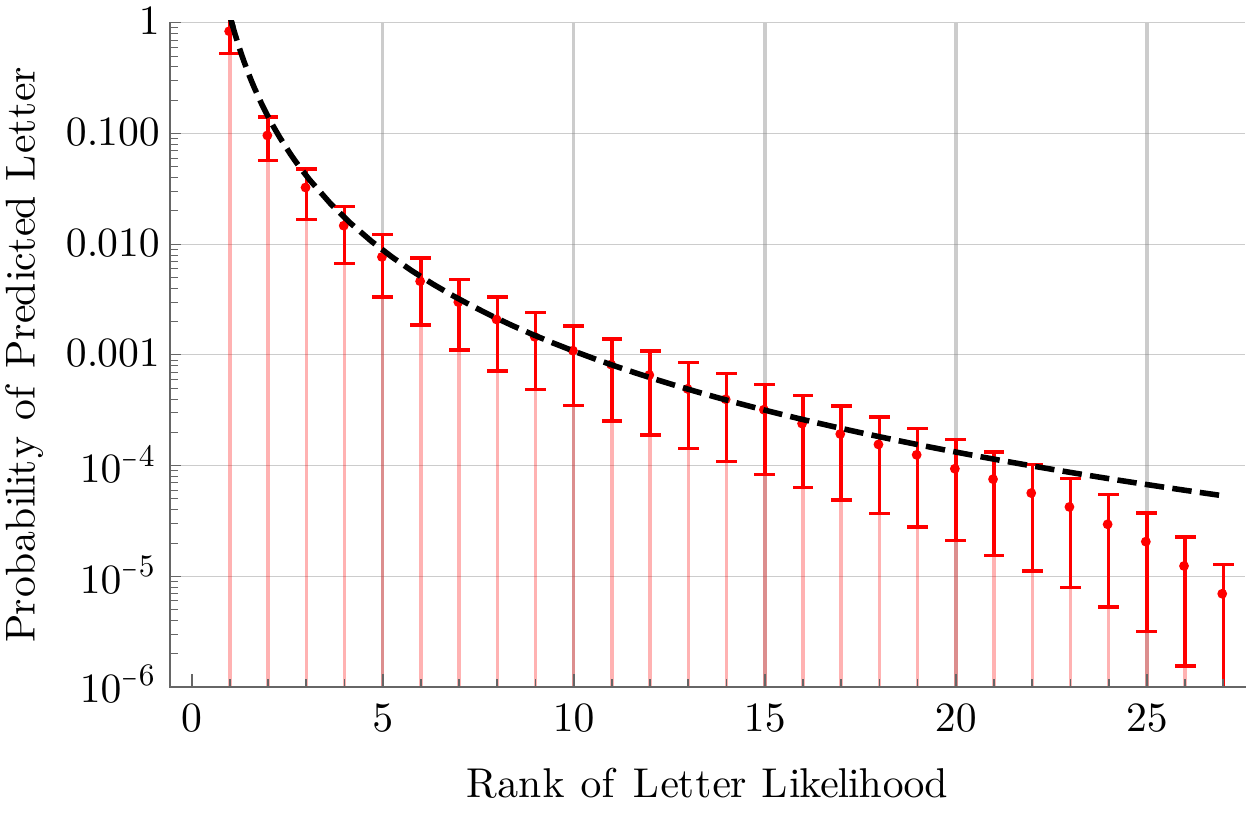}
    \caption{Log plot of the power law distribution of the output probabilities obtained from Mozilla's \emph{DeepSpeech} voice recognition LSTM on the Mozilla Common Voice verified test dataset for English \cite{Mozilla2019a}, which consists of 3995 audio samples of about ten seconds each of spoken test sentences.
    The dashed line is a fitted power law $a r^{-b}$ with parameters $a=1.2\pm0.1$ and $b=3.03\pm 0.03$.
We individually process each audio file, and capture the output after the final \texttt{Softmax} layer (\texttt{logits:0}), but before it is processed further by the greedy connectionist temporal classification (CTC beam search) implemented by \emph{DeepSpeech}.}
    \label{fig:lstpower-law}
\end{figure}

\section{Rank of Letter Likelihood for Mozilla's DeepSpeech}
DeepSpeech processes mel-frequency cepstral coefficients extracted from a sliding window of 25 miliseconds, with a stride of 20 miliseconds; for each such frame, the LSTM is invoked, and yields a distribution over the letters of the english alphabet ``a'' to ``z'', as well as a few special symbols, e.g.\ ``silence''. For the specific architecture of the LSTM we refer the reader to the original paper \cite{Hannun2014}.
Our hypothesis was that these letter probabilities follow a power-law distribution; our data (shown in \cref{fig:lstpower-law}) supports this claim.

We want to emphasize that the fact the letters a-z follow Zipf's law with respect to their occurence in English sentences (see e.g.\ \cite{Egghe2000,Piantadosi2014}) plays no role in attaining the speedup.
In addition to \cref{fig:lstpower-law}, we verified that when only collecting those output frames of \emph{DeepSpeech} where, say, ``t'' is the most likely prediction, the distribution over all letters---sorted by rank, i.e.\ sorted from most to least likely prediction---is already a power-law.
This is a feature of the output of the model, and not necessarily a property of the underlying data the model was trained on.
In our context this means that the Softmax output layer of the LSTM has to yield a power-law probability distribution.
How frequently a given letter is the most likely prediction---which is itself known to be a power-law, as mentioned---is not important.

\end{document}